\documentclass[hidelinks,onefignum,onetabnum]{siamart251216}
\usepackage{subcaption}

\usepackage{amssymb}
\usepackage{amsfonts}
\usepackage{graphicx}
\usepackage{epstopdf}
\usepackage{enumerate}
\usepackage{wrapfig}
\usepackage[rightcaption]{sidecap}
\usepackage{algorithmic}

\usepackage{hyperref}
\ifpdf
  \DeclareGraphicsExtensions{.eps,.pdf,.png,.jpg}
\else
  \DeclareGraphicsExtensions{.eps}
\fi
\usepackage[normalem]{ulem}

\newsiamremark{remark}{Remark}
\newsiamremark{hypothesis}{Hypothesis}
\crefname{hypothesis}{Hypothesis}{Hypotheses}
\newsiamthm{claim}{Claim}
\newsiamremark{fact}{Fact}
\newsiamremark{ex}{Example}
\crefname{fact}{Fact}{Facts}

\headers{Constrained Approximation for Quantum Signal Processing}{Y. Dong, J. B. Larsen, L. Lin, and R. Sarkar}

\title{Constrained Minimax Approximation for Quantum Signal Processing
\thanks{\textbf{Funding:} J.L. acknowledges support from the Department of Energy
Computational Science Graduate Fellowship under Award Number DE-SC0024386. R.S. and L.L. acknowledge support from the U.S. Department of Energy, Office of Science, Accelerated Research in Quantum Computing Centers, Quantum Utility through Advanced Computational Quantum Algorithms, grant no. DE-SC0025572. L.L. is a Simons Investigator in Mathematics.}
}

\author{Yulong Dong\thanks{Department of Electrical Engineering and Computer Science, University of Michigan, Ann Arbor, MI
  (\email{dongyl@umich.edu}).}
\and James B. Larsen\thanks{Department of Mathematics, University of Michigan, Ann Arbor, MI 
  (\email{jblarsen@umich.edu}).}
\and Lin Lin\thanks{Department of Mathematics, University of California, Berkeley, CA 
   (\email{linlin@math.berkeley.edu}).}
\and Rahul Sarkar\thanks{Department of Mathematics, University of California, Berkeley, CA 
  (\email{rsarkar@berkeley.edu}). Corresponding author.}
}

\usepackage{amsopn}

\newcommand{\RR}{\mathbb{R}}
\newcommand{\CC}{\mathbb{C}}
\newcommand{\DD}{\mathbb{D}}
\newcommand{\TT}{\mathbb{T}}

\newcommand{\calZ}{\mathcal{Z}}

\newcommand{\bc}{\mathbf{c}}
\newcommand{\bu}{\mathbf{u}}
\newcommand{\bv}{\mathbf{v}}
\newcommand{\bga}{\boldsymbol{\gamma}}

\newcommand{\hatX}{\hat{X}}
\newcommand{\hatY}{\hat{Y}}
\newcommand{\hatZ}{\hat{Z}}

\newcommand{\hzeta}{\hat{\zeta}}
\newcommand{\hbc}{\hat{\mathbf{c}}}
\newcommand{\he}{\hat{e}}

\newcommand{\inner}[1]{\left \langle #1, \Phi \right \rangle}
\newcommand{\innerx}[1]{\left \langle #1, \Phi(x) \right \rangle}

\newcommand{\norm}[2]{\left\lVert#1\right\rVert_{#2}}

\ifpdf
\hypersetup{
  pdftitle={Constrained Minimax Approximation for Quantum Signal Processing},
  pdfauthor={Yulong Dong, James B. Larsen, Lin Lin, and Rahul Sarkar}
}
\fi

\begin{document}

\maketitle

\begin{abstract}
Quantum signal processing (QSP) provides a simple and efficient framework for implementing polynomial transformations using quantum circuits. Its classical design stage leads to a constrained minimax approximation problem: find a polynomial of prescribed parity that approximates a target function uniformly on a fitting set while remaining bounded in magnitude by one on the domain $[0,1]$, which can be viewed as a semi-infinite constraint. Discretization converts the problem into a linear program, but feasibility at a set of finitely many sampled points does not ensure feasibility on the whole domain, especially when an optimal approximant reaches the boundary of the feasible set. We investigate two approaches to address this difficulty. A Remez exchange method combined with active-set constraint enforcement is efficient on many tested instances, but its stability depends on the target and problem geometry. We then introduce nonlinear Fourier retraction, which uses QSP completion and phase synthesis to turn a nearly feasible polynomial into phase factors for a feasible QSP polynomial without increasing the degree. Across representative problems, retraction largely preserves approximation accuracy and remains effective on instances where the Remez heuristic is unstable. The resulting workflow connects classical minimax approximation and semi-infinite optimization with nonlinear Fourier analysis, and is implemented in the \texttt{qsppack} software package.
\end{abstract}



\section{Introduction}
\label{sec:intro}

We address a polynomial approximation problem that arises in the design of several quantum algorithms, using the quantum signal processing (QSP) and quantum singular value transform (QSVT) protocols. For an integer $d \ge 1$, the optimization problem takes the following form of a \textit{constrained Chebyshev approximation} problem:
\vspace*{-0.3cm}
\begin{equation}
\label{eq:intro-opt-prob}
\begin{split}
\min_{\bc \in \RR^{\tilde{d}}} & \quad \sup_{x \in Y} \left \lvert \langle \bc , \Phi(x) \rangle - g(x) \right \rvert \\
\text{s.t.} & \quad -1 \le \langle \bc , \Phi(x) \rangle \le 1, \quad \forall \; x \in X,
\vspace*{-0.3cm}
\end{split}
\end{equation}
where $\tilde{d} := \lceil \frac{d+1}{2} \rceil$, $X := [0,1]$, $Y \subseteq X$ and compact, $g \in C(Y)$ is called the \textit{target function} with $C(Y)$ denoting the space of continuous functions on $Y$, and $\Phi(x) := (\phi_1(x), \dots, \phi_{\tilde{d}}(x))$. The quantities $\phi_1, \dots, \phi_{\tilde{d}}$ appearing in \cref{eq:intro-opt-prob} are polynomials of parity $d$, defined as $\phi_j := T_{2j-2}$ if $d$ is even, and $\phi_j := T_{2j-1}$ if $d$ is odd, where $T_j$ is the $j^{\text{th}}$ Chebyshev polynomial of the first kind \cite{rivlin2020chebyshev}, and the inner-product notation denotes the linear combination $\langle \bc, \Phi(x) \rangle := \sum_{j=1}^{\tilde{d}} c_j \phi_j(x)$, where $\bc := (c_1,\dots,c_{\tilde{d}})$. With these choices, the highest possible degree of the polynomial $\langle \bc , \Phi \rangle$ is $d$, which is also the degree of $\phi_{\tilde{d}}$. The constraints appearing in \cref{eq:intro-opt-prob} are called \textit{bound constraints}, and the compact subset $Y$ will be called the \textit{fitting set}. In words, the optimization problem \cref{eq:intro-opt-prob} seeks an optimal approximating polynomial of maximum degree and parity $d$, bounded in absolute value by one on $[0,1]$, such that it achieves the minimum error to $g$ in the supremum norm over the fitting set $Y$.

In many problems of practical interest in QSP, one is given a set $Y$ that takes the form of a closed interval or a disjoint union of a finite number of closed intervals, and a target function $g \in C(Y)$, and one needs to first find $P' \in \CC[x]$ such that its imaginary part $P := \Im(P')$ approximates $g$ on $Y$ (this task is the focus of this paper), followed by finding a set of phase factors $\Psi:= \left \{\psi_k \in \left(-\frac{\pi}{2},\frac{\pi}{2} \right) : k=0,1,\dots,d \right\}$ that gives rise to $P'$, according to the following $\textrm{SU}(2)$ valued QSP protocol\footnote{The QSP protocol convention used here is called the \textit{Chebyshev QSP} protocol. Two other QSP protocols, the \textit{Analytic QSP} and \textit{Laurent QSP} protocols, exist in literature, and the relationships between them can be found in \cite{laneve2025generalized}.}:
\begin{equation}
\label{eq:qsp-protocol}
U_d(\Psi,x)=
\begin{pmatrix}
    P'(x) & i Q'(x) \sqrt{1-x^2}\\
    i \overline{Q'}(x) \sqrt{1-x^2} & \overline{P'}(x)
\end{pmatrix} := e^{i \psi_0 Z} \prod_{k=1}^{d} \left( \begin{smallmatrix}
    x & i \sqrt{1-x^2} \\
    i \sqrt{1-x^2} & x
\end{smallmatrix} \right) e^{i \psi_k Z},
\vspace*{-0.3cm}
\end{equation}
defined for all $x \in [0,1]$, where $Z = \left( \begin{smallmatrix}
    1 & 0 \\
    0 & -1
\end{smallmatrix} \right)$, and $P',Q' \in \CC[x]$ have degrees at most $d$ and $d-1$, with parities $d \bmod 2$ and $(d-1) \bmod 2$ respectively~\cite[Theorem~4, Corollary~5]{Gily_n_2019}. For example, the matrix inversion problem \cite[Section~VI.B]{Martyn2021} requires a polynomial approximation to $g(x) = \frac{1}{\beta \kappa x}$ on $Y =\left[1/\kappa, 1\right]$, for some $\beta, \kappa \ge 1$. For this problem, an analytical form for the optimal approximating polynomial can be found in \cite{Berntson2024,privalov2007}, when no bound constraints are present. Usually for most other QSP problems, such closed form expressions are not known, even in the unconstrained case. An example is the uniform singular value amplification task \cite[Theorem~30]{Gily_n_2019}, useful for accelerating Hamiltonian simulation \cite{low2017hamiltoniansimulationuniformspectral}, where we have $Y = [0, \Gamma]$ for some $\Gamma \in (0,1)$, and
\vspace*{-0.3cm}
\begin{equation}
\label{eq:uniformamp}
    g(x) = \frac{x}{\beta \Gamma}, \quad x\in [0, \Gamma], \quad \beta \ge 1.
\vspace*{-0.3cm}
\end{equation}
One needs computational techniques to find an optimal approximating polynomial for this problem, even in the unconstrained case, either based on discretization, or based on Remez-type methods (see \cite{cheney1966introduction,meinardus2012approximation,remes1934procede,remez1934determination,remez1934calcul}). For both these examples, we have $\sup_{x \in Y} |g(x)| \le \frac{1}{\beta}$, and when $\beta$ and $d$ are both large, the solution to the unconstrained problem is also a solution to the constrained problem \cref{eq:intro-opt-prob}\footnote{This can be argued from the continuity of the best approximation to the unconstrained problem; see for example \cite[Theorem~3-4]{rice1964approximation}.}. Therefore, having either an analytical formula as in the case of the matrix inversion problem, or a computational method as in the case of uniform singular value amplification task, suffices to solve the original problem~\cref{eq:intro-opt-prob}. This situation is also robust to numerical inaccuracies arising from finite precision and termination conditions of numerical algorithms, as the candidate solutions do not violate the constraints in \cref{eq:intro-opt-prob}, even after accounting for numerical errors. 

However, as $\beta$ converges to one in the above examples, we reach the fully-coherent regime (defined in \Cref{ssec:best-approx}) for these problems, where any optimal polynomial approximation to \cref{eq:intro-opt-prob} approaches the boundary of the feasible set of~\cref{eq:intro-opt-prob}. In the fully-coherent regime, unconstrained optimization methods are no longer effective, and one needs to use methods that enforce the bound constraints. Such methods existing in the literature are again based on discretization (e.g. \cite{reemtsen1987discretization,reemtsen1989note}), Remez-type ideas (e.g. \cite{Chalmers1976,reemtsen1990modifications,Lai2003}), or a combination of the two within a semi-infinite programming framework \cite{hettich1986implementation}. These methods are iterative, and converge under certain conditions to an optimal approximating polynomial (though not in a finite number of iterations in general); but none of them guarantee that the successive iterates stay in the feasible set of~\cref{eq:intro-opt-prob}. When using discretization, this has already been observed for typical QSP problems of interest. A less investigated aspect is the performance of Remez-type methods for QSP problems in this regime, particularly regarding the feasibility of the iterates and issues of numerical stability, and in this work, we shed more light on this topic. Irrespective of which type of method is used to solve~\cref{eq:intro-opt-prob}, to be of further use in downstream QSP workflows, the resulting approximating polynomial obtained numerically must be corrected for residual infeasibility, necessitating good algorithms for this post-processing task.

\subsection{Related works}
\label{ssec:related-works}

The problem of polynomial approximation over Haar subspaces (see \cite{Chalmers1976} for a basis-independent definition), in the setting of $Y=X$, and with bound constraints was first considered in a series of works \cite{Taylor1968,schumaker1969approximations,taylor1969approximation}, where existence, uniqueness and equioscillation-type characterizations of the optimal approximating polynomial were derived, and an algorithm based on \cite{Moursund1968} to find the optimal approximating polynomial was presented. For the same exact setting, Remez-type algorithms were subsequently presented in \cite{Taylor1970,Gimlin1974}, and later the Haar subspace condition for these methods was relaxed in \cite{Watson1974,Chalmers1976}. The additional requirement $Y=X$ was removed in the work of \cite{reemtsen1990modifications}, where a modified first Remez algorithm was presented; however, the resulting algorithm required the solution of a constrained Chebyshev approximation problem on a finite grid in every iteration. Some Remez-type methods for this problem, relying on the Haar subspace condition, were also proposed in works related to the design of finite impulse response filters --- for example, \cite{Grenez1983} suggests a Remez-type method in the setting $Y=X$, while the works \cite{Lai2003,Lai2004} develop a Remez-exchange algorithm for the case when the fitting set $Y$ and the set where the bound constraints are imposed, are disjoint. Algorithms based on discretization of the fitting set~\cite{reemtsen1987discretization,reemtsen1989note}, or additionally the constraint set $X$, within the framework of linear semi-infinite programming (see \cite{Hettich1993, Reemtsen1998,Goberna2018,Potchinkov1997}), are also applicable to our setting and do not require the Haar subspace condition or $Y=X$ to be true.

Within the context of finding optimal approximating polynomials for QSP applications, Remez-type algorithms were suggested in \cite{Low_2016,Martyn2021}, but explicitly checking for and enforcing the feasibility of the obtained solutions were not discussed. An implementation of a Remez-type method was also presented in \cite[Appendix~E]{Dong2021}, but none of these prior works explore the performance of the algorithms in the fully-coherent regime. Optimal unconstrained  approximating polynomials for the specific matrix inversion problem (in the $d$ odd case) were considered in \cite{sunderhauf2025matrixinversionpolynomialsquantum}, where an explicit constraint enforcement procedure from \cite{Martyn2021} based on windowing is recommended. In practice, and for several problems of interest in QSP, one has good constructive approximations for polynomials satisfying the bound constraints in \cref{eq:intro-opt-prob} and achieving a target level of accuracy (typically in supremum norm), and a plethora of such results can be found in \cite{Gily_n_2019} for commonly encountered functions in QSP such as the sign, rectangle, and windowing functions, exponentials, $1/x$, trigonometric functions, and piecewise smooth functions (see also \cite{Martyn2021,low2017hamiltoniansimulationuniformspectral} for more such examples). However, these constructive approximations are not optimal.

\subsection{Contributions}
\label{ssec:contrib}
In this paper, we introduce \Cref{alg:constrained_remez_qsp} that adapts a Remez algorithm with an active-set method to handle feasibility outside the fitting set. Using several numerical experiments, we demonstrate in \Cref{ssec:remez-numerical-dems} practical performance for problems of interest and highlight a specific instance when this algorithm encounters numerical instabilities.
To address this numerically unstable instance, we introduce the nonlinear Fourier retraction algorithm (\Cref{alg:retraction}) based on the connection between QSP and nonlinear Fourier analysis. This algorithm takes as input an infeasible approximate solution to \cref{eq:intro-opt-prob} and produces a set of QSP phase factors that correspond to a feasible approximate solution. We repeat several of the experiments from \Cref{ssec:remez-numerical-dems} in \Cref{ssec:numerical-examples} to highlight the efficacy of this second algorithm.

\subsection{Notation}
\label{ssec:notation}
We introduce some common notation that will be used throughout. Matrices will always be real, except in \Cref{sec:nonlinear-retraction} where they can be complex. For a real matrix $A$, we denote its transpose as $A^\top$, while for a complex matrix $A$, we denote its Hermitian conjugate as $A^\ast$. For a vector $\bv \in \RR^n$ (or $\CC^n$), we denote its components as $v_1,\dots,v_n$ (i.e. $\bv = (v_1,\dots,v_n)$), and define $\norm{\bv}{\infty} := \max \{ |v_1|, \dots, |v_n| \}$. All row and column vectors will always be denoted by boldface letters. The natural logarithm will be denoted as log. For a complex number $z$, its real and imaginary parts will be denoted $\Re(z)$ and $\Im(z)$ respectively.

Polynomials, except possibly in \Cref{ssec:qsp-nlft} and \Cref{sec:nonlinear-retraction}, will always be real-valued polynomials. For a compact subset $Y$, if $g \in C(Y)$, we define $\norm{g}{Y} := \sup_{x \in Y} |g(x)|$. In \Cref{ssec:qsp-nlft} and \Cref{sec:nonlinear-retraction}, we denote the unit circle as $\TT$ and the open unit disk as $\DD$. For a Laurent polynomial $a(z)$, define $a^*(z) := \overline{a(\overline{z}^{-1})}$ for $z\in\CC \setminus \{0\}$, where the $\overline{\cdot}$ means the complex conjugate. Laurent polynomials will always be defined for a single variable over the field $\CC$, and have finite degree. For $p \ge 1$, the space of measurable functions on $\TT$, such that $\int_{\TT} |f(e^{i\theta})|^p \; d\theta < \infty$, is denoted $L^p(\TT)$. For a continuous function $f$ on $\TT$, we denote its supremum norm as $\norm{f}{L^\infty(\TT)}:= \sup_{z \in \TT} |f(z)|$. Products of matrices $\prod_{j=1}^{n} A_j$ will always be ordered from left to right. Other notations will be introduced along the way in the subsequent sections, when necessary.

\subsection*{Code availability}
The algorithms presented in this paper have been implemented in the Python version of \texttt{qsppack}. Full documentation with example usage can be found at \href{https://qsppack.readthedocs.io}{qsppack.readthedocs.io}.

\subsection*{On LLM use}
Large language models (LLMs) were used in writing some of the code for performing the numerical experiments, but all such code was manually reviewed for correctness. LLMs were also consulted asynchronously to identify typos, errors, or redundancies in the manuscript; and any suggestions were reviewed and implemented serially and by-hand if correct.

\section{Preliminaries}
\label{sec:prelim}

In this section, we first introduce some terminology characterizing the optimal solution to \cref{eq:intro-opt-prob} in \Cref{ssec:best-approx}. In \Cref{ssec:finite-discretization}, we describe discretization-based linear programming strategies to approximately solve \cref{eq:intro-opt-prob}. Next, in \Cref{ssec:qsp-nlft}, we discuss the connection between QSP and nonlinear Fourier analysis and give context as to why we are interested in solving the optimization problem \cref{eq:intro-opt-prob}. We then briefly state the connections between QSP. In \Cref{ssec:expt}, we provide the relevant QSP applications that will serve as the numerical experiments.

\subsection{Best approximation}
\label{ssec:best-approx}

Let $\calZ : \RR^{\tilde{d}} \ni \bc \mapsto \langle \bc, \Phi\rangle \in \RR[x]$ be the linear map that maps the unknowns of the optimization problem \cref{eq:intro-opt-prob} to a subspace of polynomials of dimension $\tilde{d}$. The set of feasible solutions to \cref{eq:intro-opt-prob} is denoted as
\vspace*{-0.1cm}
\begin{equation}
\label{eq:feasible-sol}
F := \{\bc \in \RR^{\tilde{d}} : -1 \le \langle \bc , \Phi(x) \rangle \le 1, \quad \forall \; x \in X\},
\end{equation}
and we note that $F$ is non-empty, as $0 \in F$, and it has a non-empty interior. It is also clear that $F$ is a closed, convex subset\footnote{Convexity of $F$ follows from linearity of the map $\calZ$. To see that $F$ is closed in $\RR^{\tilde{d}}$, consider a convergent sequence $\bc_j \rightarrow \bc$ in $\RR^{\tilde{d}}$, such that each $\bc_j \in F$. Then we have by continuity of $\calZ$ that $\norm{\inner{\bc}}{X} \le 1$, and therefore $\bc \in F$.} of $\RR^{\tilde{d}}$. The boundary of $F$ will be denoted $\partial F$. The image of $F$ under the map $\calZ$, i.e. $\calZ(F)$ will be called the set of \textit{feasible polynomials}, while we say that a polynomial $P$ is an \textit{infeasible polynomial} if $P \not \in \calZ(F)$. Note that the specific choice of $\phi_1,\dots,\phi_{\tilde{d}}$ in \cref{eq:intro-opt-prob} is for convenience. One could make other choices too: for example, when $d$ is odd, we have $\text{span} \{\phi_1,\dots,\phi_{\tilde{d}}\} = \text{span} \{x, x^3,\dots,x^d\}$, while for $d$ even, we have $\text{span} \{\phi_1,\dots,\phi_{\tilde{d}}\} = \text{span} \{1, x^2,\dots,x^d\}$, and one could have also chosen $\Phi(x) := (x, x^3, \dots, x^d)$ and $\Phi(x) := (1, x^2, \dots, x^d)$, in the odd and even cases respectively. 

We will say that $\bc^\ast \in F$ is a \textit{best approximation} to $g \in C(Y)$, or interchangeably an \textit{optimal solution} to the optimization problem \cref{eq:intro-opt-prob}, if and only if
\vspace*{-0.1cm}
\begin{equation}
\label{eq:best-approx-def}
\norm{\langle \bc^\ast , \Phi \rangle - g}{Y} \le \norm{\langle \bu , \Phi\rangle - g}{Y}, \quad \forall \; \bu \in F,
\vspace*{-0.1cm}
\end{equation}
and we define the quantity $e^\ast := \norm{\langle \bc^\ast , \Phi \rangle - g}{Y}$ to be the \textit{best error}. The corresponding polynomial $\calZ(\bc^\ast)$ will be called a \textit{best approximation polynomial}. The existence of an optimal solution $\bc^\ast$ is guaranteed by standard compactness arguments, and directly follows from the existence of an optimal solution to problem~\cref{eq:lsip-prob} introduced below, which is equivalent to \cref{eq:intro-opt-prob} when one sets $\hatX = X$ and $\hatY = Y$.

A best approximation polynomial is not necessarily unique. In fact, there can be infinitely many best approximation polynomials for the optimization problem \cref{eq:intro-opt-prob}, and it depends in general on all the problem parameters involved, such as the fitting set $Y$, the target function $g$, as well as the maximum degree of the approximating polynomial $d$. We provide some examples in \Cref{appssec:additional}. Finally, an instance of problem~\cref{eq:intro-opt-prob} is said to belong to the \textit{fully-coherent} regime if it has a best approximation $\bc^\ast \in \partial F$, or equivalently if $\norm{\inner{\bc^\ast}}{X} = 1$. Usually this happens when the target function satisfies $\norm{g}{Y} = 1$.

\subsection{Discretization strategies}
\label{ssec:finite-discretization}

Let us first introduce closed subsets $\hatX \subseteq X$ and $\hatY \subseteq Y$, and consider the following \textit{linear semi-infinite program} (LSIP):
\vspace*{-0.2cm}
\begin{equation}
\label{eq:lsip-prob}
\begin{split}
\min_{(\bc,e) \in \RR^{\tilde{d}+1}} & \quad \quad e \\
\text{s.t.}  & \quad -e \le \langle \bc , \Phi(x) \rangle - g(x )\le e, \quad \forall \; x \in \hatY, \\
& \quad -1 \le \langle \bc , \Phi(x) \rangle \le 1, \quad \forall \; x \in \hatX,
\vspace*{-0.3cm}
\end{split}
\end{equation}
which is exactly equivalent to the optimization problem of minimizing $\norm{\inner{\bc} - g}{\hatY}$, subject to the constraints $\norm{\inner{\bc}}{\hatX} \le 1$. Thus when $\hatX = X$ and $\hatY = Y$, the optimization problem~\cref{eq:intro-opt-prob} and the LSIP~\cref{eq:lsip-prob} are equivalent. We will denote the feasible set for \cref{eq:lsip-prob} as $F_{(\hatX, \hatY)}$, and denote an optimal solution to it as $\hzeta^\ast := (\hbc^\ast, \he^\ast)$. An optimal solution to \cref{eq:lsip-prob} always exists, a fact we record in the next lemma, and whose proof is delegated to \Cref{appssec:additional}. Note that the feasible set $F_{(\hatX,\hatY)}$ is not empty, as $(0, \norm{g}{\hatY}) \in F_{(\hatX,\hatY)}$.

\begin{lemma}[Existence]
\label{lem:existence-optimal-sol}
The LSIP \cref{eq:lsip-prob} always has an optimal solution $\hzeta^\ast$.
\end{lemma}

We give two additional properties of the LSIP \cref{eq:lsip-prob} as lemmas below. The proofs are obvious and skipped.

\begin{lemma}[Non-decreasing best error]
\label{lem:lsip-increasing}
Let $\hatX_1, \hatX_2$, $\hatY_1$, $\hatY_2$ be compact subsets such that $\hatX_1 \subseteq \hatX_2 \subseteq X$ and $\hatY_1 \subseteq \hatY_2 \subseteq Y$. Now consider two cases of the LSIP \cref{eq:lsip-prob} with (a) $(\hatX,\hatY) = (\hatX_1,\hatY_1)$, and (b) $(\hatX,\hatY) = (\hatX_2,\hatY_2)$, and suppose $\hzeta^\ast_1 = (\hbc^\ast_1, \he^\ast_1)$ and $\hzeta^\ast_2 = (\hbc^\ast_2, \he^\ast_2)$ be an optimal solution for problem (a) and (b) respectively. Then $\he^\ast_1 \le \he^\ast_2 \le e^\ast$.
\end{lemma}
\begin{lemma}[Optimal LSIP -- optimal original]
\label{lem:lsip-to-full-optimality}
Consider the LSIP \cref{eq:lsip-prob} and let $\hzeta^\ast = (\hbc^\ast, \he^\ast)$ be an optimal solution. Then $\hbc^\ast$ is an optimal solution to the optimization problem \cref{eq:intro-opt-prob} with best error $\he^\ast$ if and only if $\hbc^\ast \in F$ and $\norm{\inner{\hbc^\ast} - g}{Y} = \he^\ast$.
\end{lemma}

The optimization problem \cref{eq:lsip-prob} still has an infinite number of constraints when either $\hatX$ or $\hatY$ is infinite. In order to solve it on a computer, one chooses $\hatX$ and $\hatY$ to be finite. This step is well-known in the literature as \textit{discretization}, which converts problem~\cref{eq:lsip-prob} into a linear program (LP). As \Cref{lem:lsip-increasing} shows, the optimal solution to \cref{eq:lsip-prob} is then always a lower bound to the best error $e^\ast$ of the original optimization problem~\cref{eq:intro-opt-prob}, and one expects that as $\hatX$ and $\hatY$ become more dense in $X$ and $Y$ respectively, the optimal solution to \cref{eq:lsip-prob} should converge to $e^\ast$ (see \cite[Section~6]{Shapiro2009}). 

\begin{SCfigure}
\includegraphics[width=0.55\linewidth]{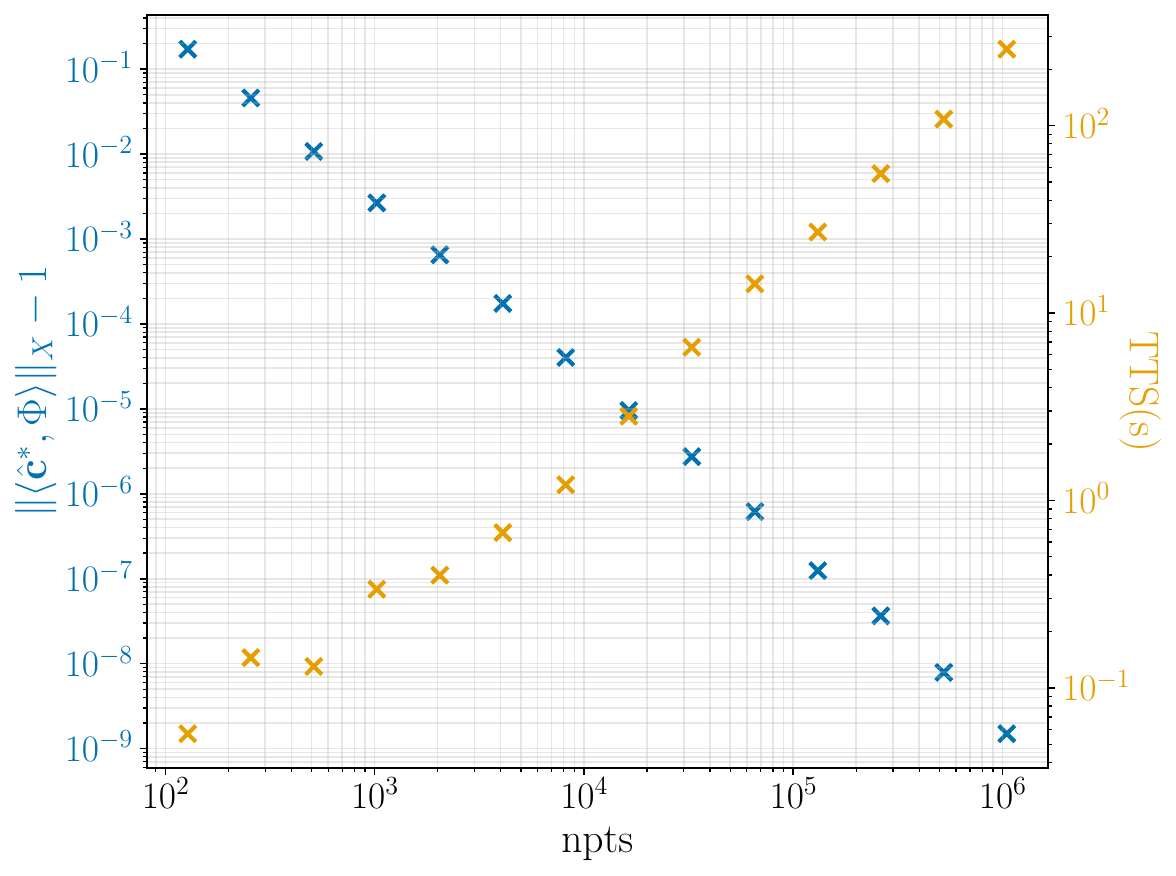}
\caption{Uniform singular value amplification example \cref{eq:uniformamp} with $\beta=1, \Gamma=0.2$, and polynomial degree $d=101$. The plot shows the maximum constraint violation $\norm{\inner{\hbc^\ast}}{X} - 1$ in blue, and time-to-solution (TTS) in maize required by \texttt{cvxpy} (using \texttt{CLARABEL} solver \cite{Clarabel_2024}) to solve the LP in \cref{eq:lsip-prob}, versus the number of discretization points $\texttt{npts} = |\hatX|$.}
\label{fig:constraintviol}
\vspace*{-20pt}
\end{SCfigure}

In prior works, one can also find results (see \cite[Theorem~3.2, Theorem~3.3]{Shapiro2009}, \cite[Theorem~4.1]{borwein1981direct}) which imply the existence of finite sets $\hatX$ and $\hatY$, with $|\hatX| + |\hatY| \le \tilde{d} + 1$, such that the minimum value of the LP \cref{eq:lsip-prob} equals the best error of \cref{eq:intro-opt-prob}. Note that these results are not constructive, and there are two difficulties in applying this existence result in actual practical instances of \cref{eq:intro-opt-prob}: (i) there is no known algorithm to compute the finite sets $\hatX$ and $\hatY$, and (ii) even if one can find the finite sets $\hatX$ and $\hatY$, the theorems do not guarantee that the $\hbc^\ast$ obtained from the optimal solution $\hzeta^\ast = (\hbc^\ast, \he^\ast)$ to the resulting LP, is also a feasible solution of \cref{eq:intro-opt-prob}. In order to get around the first difficulty, as a practical strategy, one chooses $\hatX$ and $\hatY$ such that $|\hatX|, |\hatY|$ are large and sufficiently dense in $X, Y$ respectively. As a result, for the optimal solution $\hzeta^\ast = (\hbc^\ast, \he^\ast)$ to the LP, we first have that $|e^\ast - \he^\ast|$ is small, and moreover if $\hbc^\ast \not \in F$, then $\norm{\inner{\hbc^\ast}}{X} - 1$ is also small. This issue of infeasibility, i.e. $\hbc^\ast \not \in F$, is especially pronounced when approaching the fully-coherent regime, which is demonstrated on the uniform singular value amplification task from \cref{eq:uniformamp} in \Cref{fig:constraintviol}. We choose $\beta = 1$, $\Gamma = 0.2$, and $d=101$, and then solve instances of the resulting LP for different values of $|\hatX|$ (the LP setup follows the same details as mentioned in the beginning of \Cref{ssec:numerical-examples}). We see that while the maximum constraint violation measured by $\norm{\inner{\hbc^\ast}}{X} - 1$ decreases as $|\hatX|$ increases, in every single case the optimal solution $\hbc^\ast$ is infeasible for the original problem~\cref{eq:intro-opt-prob}, no matter how large a value of $|\hatX|$ is chosen. This necessitates the need for a numerical method that takes a solution that is \textit{nearly feasible}\footnote{An infeasible solution $\bc$ to \cref{eq:intro-opt-prob} is understood to be nearly feasible if $\norm{\inner{\bc}}{X} \le 1 + \delta$, for some small $\delta > 0$.}, e.g. produced by solving the LP, or via a different method such as those of Remez-type, and gives a feasible solution to \cref{eq:intro-opt-prob} --- one simple method to do this is based on scaling the resulting polynomial, but we present another method in this paper in \Cref{sec:nonlinear-retraction} that performs better on many practical QSP problems of interest.

\subsubsection{Certification of feasibility}
\label{sssec:certified-feasibility}
As the above discussion suggests, when solving problem~\cref{eq:intro-opt-prob}, and especially in the fully-coherent regime, it is important to be able to test numerically whether a candidate solution $\bc \in \RR^{\tilde{d}}$ to the problem is feasible or not. The method that we will employ in this paper is based on finding all the roots of the derivative of the polynomial $\inner{\bc}$ in the interval $X$, which isolates all the local extremas (minimas and maximas), and then checks whether $1 \le \inner{\bc}(x) \le 1$, for all $x$ belonging to this extremal set, and at the endpoints $x=0,1$. Consequently, in case of constraint violations, this also gives us the maximum constraint violation $\norm{\inner{\bc}}{X}-1$. Stable implementations of this algorithm are provided by software libraries like \texttt{chebfun} \cite{driscoll2014chebfun}; we have also included an implementation as \texttt{check\_feasibility} in \texttt{qsppack.utils}.

\subsection{QSP and the nonlinear Fourier transform}
\label{ssec:qsp-nlft}

As mentioned in \Cref{sec:intro}, when using the Chebyshev QSP protocol in designing quantum algorithms, one needs to first approximate the target function $g \in C(Y)$, with $\norm{g}{Y} \le 1$, by a real polynomial $P$ with parity $d \bmod{2}$, and this is called the \textit{approximation} step. Since both sides of \cref{eq:qsp-protocol} are elements of $\textrm{SU}(2)$, we have $|P'(x)|^2 + (1-x^2)|Q'(x)^2|=1$ for all $x \in X$, and, in particular, this implies $\norm{\Re(P')}{X} = \norm{P}{X} \le 1$. However, before one can proceed to the next step of finding a set of phase factors $\Psi$, one must find polynomials $P', Q' \in \mathbb{C}[x]$ given $\Im(P')= P \in \mathbb{R}[x]$. This is known as the \textit{completion problem}. Polynomial factorization-based methods have been suggested in the literature \cite{Gily_n_2019, Haah2019} for solving the completion problem, but they all rely on root-finding, which is numerically unstable.

Recently, however, a correspondence between QSP and nonlinear Fourier analysis was discovered in \cite{Alexis2024,Alexis2025,ni2025inversenonlinearfastfourier,laneve2025generalized}. The $\mathrm{SU}(2)$ \textit{nonlinear Fourier transform} (NLFT) maps a sequence $\bga = \left\{ \gamma_k \in \mathbb{C} : k = 0, \dots, d\right\}$ to $\overbrace{\boldsymbol{\gamma}} \in \mathrm{SU}(2)$, defined by
\vspace*{-0.2cm}
\begin{equation}
\label{eq:nlft}
    \overbrace{\boldsymbol{\gamma}}(z) := \begin{pmatrix}
        a(z) & b(z) \\ -b^*(z) & a^*(z)
    \end{pmatrix} = \prod_{k=0}^d \frac{1}{\sqrt{1+|\gamma_k|^2}}\begin{pmatrix}
        1 & \gamma_k z^k \\ -\overline{\gamma_k}z^{-k} & 1
    \end{pmatrix},
    \vspace*{-0.2cm}
\end{equation}
where $b(z)$ and $a^\ast(z)$ are complex polynomials\footnote{Here $z$ is a complex variable. Alternatively, $b(z), a^\ast(z)$ can also be thought of as functions on $\TT$, and analytic continuation uniquely defines them everywhere in $\CC$.} of degree at most $d$ (see \cite{Alexis2024,ni2025inversenonlinearfastfourier} for detailed properties of the NLFT). The NLFT map is injective, so given polynomial pairs $(b(z), a^\ast(z))$ in the image of the map, the inverse map of finding $\bga$ is well-defined and is called the \textit{inverse nonlinear Fourier transform}. As explained in \cite[Lemma~3.1]{ni2025inversenonlinearfastfourier}, the NLFT allows us to relate $U_d(\Psi,x)$ of \cref{eq:qsp-protocol} to the NLFT of a sequence $\boldsymbol{\gamma}$ defined by $\gamma_k = \tan \psi_k$ for $k=0,\dots,d$ via the identity\footnote{We call this the \textit{QSP-NLFT correspondence}.}
\vspace*{-0.1cm}
\begin{equation}
\label{eq:HUdH-lemma-nonsym}
\begin{pmatrix}
    1&\\&i
\end{pmatrix} H U_d(\Psi,\cos\theta) H \begin{pmatrix}
    1&\\&-i
\end{pmatrix} = \overbrace{\bga}( e^{2i\theta}) 
\begin{pmatrix}
    e^{i d \theta} & 0 \\
    0 & e^{-i d \theta}
\end{pmatrix}, \quad x = \cos \theta,
\end{equation}
where $H := \frac{1}{\sqrt{2}}  \left(\begin{smallmatrix}
    1 & 1 \\1 & -1
\end{smallmatrix}\right)$. Given $P$, a choice of $b \in \mathbb{C}[z]$ such that \cref{eq:HUdH-lemma-nonsym} holds is $b(e^{2i \theta}) := e^{i d \theta} P(\cos\theta)$ (see \cite[Section~3.3]{ni2025inversenonlinearfastfourier})\footnote{This choice also implies that $\norm{b}{L^\infty(\TT)} = \norm{P}{X}$, and the maximum modulus principle then gives $\sup_{z \in \overline{\DD}} |b(z)| = \norm{P}{X} \le 1$.}, and this mapping between the Chebyshev coefficients of $P$ and the monomial coefficients of $b$ is made explicit in \Cref{appssec:coefmap}. Since $P \in \mathbb{R}[x]$ in our context, this mapping implies that $b$ has real monomial coefficients as well.

Given a polynomial $b(z)$ with $\norm{b}{L^\infty(\TT)} \le 1$, there exists a unique polynomial $a^\ast(z)$ such that $\left(\begin{smallmatrix} a(z) & b(z) \\ -b^\ast(z) & a^\ast(z)
\end{smallmatrix} \right)$ is the NLFT of some $\bga \in \CC^{d+1}$, and $a^\ast$ is an \textit{outer polynomial}, meaning that it has no zeros in $\DD$ \cite[Lemma~A.2]{ni2025inversenonlinearfastfourier}. Under the stricter assumption that $\norm{b}{L^\infty(\TT)} \le 1 - \xi$ for some $0 < \xi < 1$, it is shown in \cite{Alexis2025} that there is a simple method called the \textit{Weiss algorithm} that finds this unique outer $a^\ast$, and the algorithm is numerically stable. In our QSP context, this provides a solution to the completion problem when $\norm{P}{X} \le 1 - \xi$ via the QSP-NLFT correspondence \eqref{eq:HUdH-lemma-nonsym}, completely avoiding root-finding procedures. Furthermore, it is shown in \cite{ni2025inversenonlinearfastfourier} that, under these same assumptions on $b$ and $a^\ast$, two algorithms for computing the inverse NLFT --- ( a) \textit{layer stripping}, with runtime $\mathcal{O}(d^2)$, and (b) the \textit{inverse nonlinear fast Fourier transform} (INFFT), with runtime $\mathcal{O}(d \log^2(d))$ --- are also numerically stable. Thus, once we obtain the corresponding $\bga$, the QSP-NLFT correspondence can be used to recover a set of phase factors $\Psi$. Note that the bit requirements for all these algorithms scale as $\text{polylog}(\xi)$ in the $\xi$ parameter.

As the parameter $\xi \rightarrow 0$, we approach the fully coherent regime, and the numerical stability of the Weiss algorithm and the two inverse NLFT algorithms is no longer guaranteed. In this case, finding $a^\ast(z)$ requires root-finding methods using extended-precision arithmetic. This extended-precision requirement also applies to the subsequent inverse NLFT algorithms and appears in prior work for solving the completion and phase-factor-finding steps of QSP without going via the NLFT route. However, the recent work \cite{DongLinNiEtAl2024_newton} suggests an alternative iterative method based on Newton iterations that directly finds $\Psi$ without needing to solve the completion problem. This method behaves well in the fully coherent regime and does not require extended-precision arithmetic. All of these methods for dealing with the fully-coherent regime implicitly require $\norm{P}{X} = \norm{b}{L^\infty(\TT)} \le 1$. However, as we discussed previously in \Cref{ssec:finite-discretization}, this assumption often fails when $P$ is obtained by solving the corresponding LP \cref{eq:lsip-prob}. When $\norm{P}{X} > 1$, but say $\norm{P}{X} - 1$ is small, i.e. we have a nearly feasible solution to \cref{eq:intro-opt-prob} that very slightly violates the bound constraints, one of our key numerical observations is that combining the Weiss (with slight modifications) and inverse NLFT algorithms leads to an efficient heuristic to compute $\Psi$ for many problems of interest in QSP, even though the strict assumptions guaranteeing their numerical stability do not hold (see \Cref{sec:nonlinear-retraction}).

\subsection{Candidate QSP problems for numerics}
\label{ssec:expt}
Three important and relevant QSP problems are chosen for the numerical demonstrations of the proposed methods in \Cref{ssec:remez-numerical-dems,ssec:numerical-examples}, and \Cref{appssec:retract_exp_matrix,appssec:retract_exp_threshold}, which we list below.
\begin{enumerate}[(i)]
\item \emph{Threshold projector (TP)}. This specific example is discussed in \cite[Section~4.1]{Gily_n_2019} and has use for singular value discrimination and principal component regression. For this problem, the target function takes the following form
\vspace*{-0.3cm}
\begin{equation}
\label{eq:threshtarget}
    g_{\strut \text{TP}}(x) := \begin{cases}
        1 & x \in [0,0.5-\delta] \\
        0 & x \in [0.5+\delta,1] \,,
    \end{cases}
\vspace*{-0.3cm}
\end{equation}
with $Y=[0,0.5-\delta]\cup [0.5+\delta,1]$, for some $0 < \delta < 0.5$. It has been used for ground state preparation and energy estimation \cite{Dong2021} and was even used recently as the backbone of quantum oracle sketching \cite{zhao2026exponentialquantumadvantageprocessing} and for Gibbs sampling \cite{leng2026acceleratingquantumgibbssampling}. We will use $\delta = 0.05$ in our experiments.

\item \emph{Matrix inversion (MI)}. We already encountered this example in \Cref{sec:intro}, and here we set $\beta = 1$, that gives rise to the target function $g_{\strut \text{MI}}(x)=\frac{1}{\kappa x}$ on the interval $Y=\left[\frac{1}{\kappa},1 \right]$. We will use $\kappa=10$. Recently, this example has been used to develop an optimal quantum linear solver with the smallest known constant factor guarantees \cite{costa2026constantfactoranalysisoptimal}.

\item \emph{Uniform singular value amplification (USVA)}. This example was also introduced in \Cref{sec:intro}, specifically \cref{eq:uniformamp}, and here we set $\beta=1$, giving rise to the target function $g_{\strut \text{SV}}(x)=\frac{x}{\Gamma}$ on the fitting set $Y=[0, \Gamma]$. We will use $\Gamma=0.2$.
\end{enumerate}
These examples are chosen because they correspond to the fully-coherent regime, which is the focus of this work.

\section{A Remez and active-set based algorithm}
\label{sec:remez}

In this section, we develop a heuristic algorithm inspired by the Remez-exchange method \cite{Chalmers1976,Lai2003,Lai2004} and the active-set method \cite{murty1988linear}, which serves as a baseline for comparison. We test this algorithm on several representative QSP approximation problems, that we mentioned in \Cref{ssec:expt}. The numerical results show that, while the method can produce high-quality approximations in favorable cases, it may also become suboptimal or numerically unstable in more delicate regimes. These observations highlight intrinsic limitations of this Remez- and active-set-based approach and further motivate the method developed in \Cref{sec:nonlinear-retraction}. Remez-based methods require equioscillation of the best approximating polynomial \cite{remez1934determination,reemtsen1990modifications}, which we assume holds for \cref{eq:intro-opt-prob} as well.

\subsection{Algorithm}
Consider the approximation problem from \eqref{eq:intro-opt-prob} and \Cref{ssec:qsp-nlft} with a target function $g$, using a QSP-based polynomial class of degree $d$, in which the polynomial takes the form $f(x) = \sum_{j = 0}^d \chi_j T_j(x)$ with $\chi_j = 0$ for all $j \not\equiv d \pmod{2}$. Note that the $\chi_j$ coefficients are the same as the $c_{j}$ coefficients from \eqref{eq:intro-opt-prob} but with the zeros included, so that $j$ corresponds to the degree of the respective Chebyshev polynomial.  
Employing the mapping $x = \cos(\omega)$, this approximation problem can be written in the $\omega$-space with
\vspace*{-0.5cm}
\begin{equation}
    \tilde{g}(\omega) := g(\cos(\omega)), \qquad \tilde{f}(\omega) := f(\cos(\omega)) = \sum_{j = 0}^d \chi_j \cos(j \omega).
\vspace*{-0.3cm}
\end{equation}
This choice improves numerical implementation and stability. For example, equally spaced points in $\omega$-space correspond to Chebyshev nodes, which typically lead to numerically stable function evaluation and transformations. Due to the parity constraint, it suffices to consider the half interval $\Omega := [0, \pi/2] = \arccos([0, 1])$. The structure of the approximation problem can be divided into two parts: (i) the approximation target set $\Omega_0 := \arccos(Y)$, on which we minimize the approximation error to the target, and (ii) the boundedness constraint set $\Omega_c := \Omega \setminus \Omega_0$, on which we ensure the approximation is bounded between $\pm 1$. Note that if the approximation error is small, then the approximation polynomial oscillates around the target function, so boundedness is approximately satisfied on $\Omega_0$. As a remark, we treat the problem as a surrogate by relaxing the bound constraint on $\Omega_0$, with feasibility (global boundedness) restored later via a scaling procedure in \cref{eqn:scaling-Remez,eqn:scaling-Remez-error}.

When the constraint set is empty, namely $\Omega_c = \emptyset$, the best approximation can be solved by the Remez exchange method \cite{Chalmers1976}. Since the best approximation error is equioscillating on the approximation interval, the Remez exchange method maintains a set of extremal points with alternating signs together with the ripple amplitude. This gives an iterative method to update these error characteristics and the polynomial coefficients at each step. When inequality constraints are imposed and $\Omega_c$ is disjoint from $\Omega_0$, a modified Remez algorithm was proposed in \cite{Lai2003,Lai2004} by maintaining the equioscillation ripple on $\Omega_0$ together with the violation points of the inequality constraint on $\Omega_c$. This method does not apply directly to our case, since the required disjointness condition does not hold.

To numerically solve the best approximation polynomial in our case, which involves approximation on a subinterval together with global boundedness, we propose a method inspired by the modified Remez exchange method and the active-set method. The idea is to maintain a set of extremal points on $\Omega_0$ together with a set of active points that violate the inequality constraints on $\Omega_c$. At each iteration, we update the approximation as follows.

We first handle the boundedness constraint $-1 \le \tilde{f}(\omega) \le 1$ on $\Omega_c$. By evaluating the first-order condition $\tilde{f}^\prime(\omega) = 0$, we obtain a set of extremal points, together with the endpoints of the subintervals. These points represent the locations where the approximation polynomial may locally attain its largest amplitude. We then examine the constraint violation at these candidate points and define the following active set:
\vspace*{-0.1cm}
\begin{equation}
    \mathcal{S}_c = \{ \omega \in \Omega_c : \tilde{f}^\prime(\omega) = 0 \text{ or } \omega \in \partial \Omega_c,\ |\tilde{f}(\omega)| > 1  \}.
\vspace*{-0.2cm}
\end{equation}
At these violation points, we impose equality constraints by pinning the updated approximation polynomial to the nearest bound:
\vspace*{-0.3cm}
\begin{equation}\label{eqn:ineq_constraint}
    \tilde{f}^\star(\omega_j^c) = v(\omega_j^c) := \left\{ \begin{array}{ll}
        1 & \text{when } \tilde{f}(\omega_j^c) > 1 \\
        -1 & \text{when } \tilde{f}(\omega_j^c) < -1
    \end{array} \right., \ j = 1, \cdots, |\mathcal{S}_c|.
\vspace*{-0.3cm}
\end{equation}
We then handle the approximation error on the subinterval $\Omega_0$ by maintaining an equioscillation ripple. The maximum amplitude of the error is attained at points where the approximation error $e(\omega) := \tilde{f}(\omega) - \tilde{g}(\omega)$ satisfies the first-order condition, together with the endpoints. We collect these points into
\vspace*{-0.2cm}
\begin{equation}
    \mathcal{S}_0 = \{\omega \in \Omega_0 : e^\prime(\omega) = 0 \text{ or } \omega \in \partial \Omega_0\}.
\vspace*{-0.2cm}
\end{equation}
In the minimax regime, the approximation error should alternate in sign across these ripple points. Sorting these points from small to large as $\mathcal{S}_0 = \{\omega^0_j : j = 1, \cdots, |\mathcal{S}_0|\}$, we impose the alternating ripple condition
\vspace*{-0.1cm}
\begin{equation}\label{eqn:eq_constraint}
    \tilde{f}^\star(\omega_j^0) + (-1)^j \Delta = \tilde{g}(\omega_j^0), \ j = 1, \cdots, |\mathcal{S}_0|.
\vspace*{-0.1cm}
\end{equation}
We see that \cref{eqn:ineq_constraint,eqn:eq_constraint} form a system of linear equations for the $m = \lceil (d+1) / 2 \rceil$ coefficients and the ripple width parameter $\Delta$, where $m$ is the number of available cosine coefficients. We can then solve for the updated approximation polynomial $\tilde{f}^\star(\omega)$. This system may be overdetermined when $|\mathcal{S}_0| + |\mathcal{S}_c| > m + 1$. In practice, we select $m_\text{alt} = m+1 - |\mathcal{S}_c|$ points with the largest error amplitudes from $\mathcal{S}_0$ to form the update system, and iterate this process as an inner loop.

Note that the treatment on $\Omega_0$ follows the spirit of the Remez exchange method, while the treatment on $\Omega_c$ follows the spirit of the active-set method. By iterating this procedure, we obtain an approximation polynomial $\tilde{f}^\text{raw}(\omega)$. This function may still violate the boundedness condition on the full interval. Therefore, in the final step, we evaluate its maximal amplitude by solving the first-order condition. This gives the maximal amplitude $\|\tilde{f}^\text{raw}\|_\infty$. We then define the scaling constant by $A^\text{raw} := \max\{1, \|\tilde{f}^\text{raw}\|_\infty\}$. The output of the method is the scaled approximation polynomial
\vspace*{-0.5cm}
\begin{equation}\label{eqn:scaling-Remez}
    f^\text{final}(x) = \tilde{f}^\text{final}(\arccos(x)) = \frac{\tilde{f}^\text{raw}(\arccos(x))}{A^\text{raw}}.
\vspace*{-0.3cm}
\end{equation}
Now suppose the approximation error of the raw polynomial is
\vspace*{-0.3cm}
\begin{equation*}
    \epsilon^\text{raw} := \max_{\omega \in \Omega_0} |\tilde{f}^\text{raw}(\omega) - \tilde{g}(\omega)|,
\vspace*{-0.1cm}
\end{equation*}
and the overshoot error is $\epsilon^\text{amp} := A^\text{raw} - 1$. By the triangle inequality, we obtain
\vspace*{-0.1cm}
\begin{equation}\label{eqn:scaling-Remez-error}
    \max_{x \in Y} | f^\text{final}(x) - g(x) | \le \epsilon^\text{amp} + \epsilon^\text{raw}.
\vspace*{-0.1cm}
\end{equation}
Hence, the final output preserves the raw approximation quality only when the overshoot error is smaller than, or at least comparable to, the raw approximation error.

\begin{algorithm}
\caption{A Remez and active-set based algorithm for QSP polynomial approximation}
\label{alg:constrained_remez_qsp}
\textbf{Input:} Degree $d$, target set $Y \subset [0,1]$, target function $g(x)$. \\
\textbf{Output:} Approximation polynomial $f^\text{final}(x)=\sum_{j=0}^{d} \chi_j T_j(x)$ with $c_j=0$ for all $j \not\equiv d \pmod{2}$.
\begin{algorithmic}[1]
    \STATE Set $\Omega_0 := \arccos(Y),\ \Omega_c := [0,\pi/2]\setminus \Omega_0$ and initialize $\tilde f(\omega) = \sum_{j=0}^d \chi_j \cos(j \omega)$ with $\chi_j=0$ for all $j \not\equiv d \pmod{2}$.
    \FOR{$k=1,2,\dots$ until convergence}
        \STATE Find candidate constraint extrema on $\Omega_c$ from $\tilde f'(\omega)=0$ together with $\partial \Omega_c$
        \STATE Form the active set $\mathcal S_c$ from points with $|\tilde f(\omega)| > 1$
        \STATE Pin each $\omega_j^c \in \mathcal S_c$ to the nearest bound $\tilde f^\star(\omega_j^c) = v(\omega_j^c)\in\{\pm 1\}$
        \STATE Find candidate error extrema on $\Omega_0$ from $e'(\omega)=0$, where $e(\omega)=\tilde f(\omega)-\tilde g(\omega)$, together with $\partial \Omega_0$
        \STATE Select $m_\mathrm{alt}=m+1-|\mathcal S_c|$ points with the largest error amplitudes from $\mathcal S_0$
        \STATE Solve the linear system given by the active constraints $\tilde f^\star(\omega_j^c)=v(\omega_j^c)$ and the alternating ripple conditions $\tilde f^\star(\omega_j^0)+(-1)^j\Delta=\tilde g(\omega_j^0)$
        \STATE Update $\tilde f \leftarrow \tilde f^\star$
        \STATE Stop if the constraint violation changes by less than the prescribed tolerance
    \ENDFOR
    \STATE Denote the output by $\tilde f^\text{raw}$
    \STATE Compute $A^\text{raw} := \max\{1,\|\tilde f^\text{raw}\|_\infty\}$
    \STATE Set $\tilde f^\text{final}(\omega) := \tilde f^\text{raw}(\omega)/A^\text{raw} = \sum_j \chi_j^\text{raw}/A^\text{raw} \cos(j \omega)$
    \RETURN $f^\text{final}(x) = \sum_j c_j^\text{raw}/A^\text{raw} T_j(x)$
\end{algorithmic}
\end{algorithm}

\subsection{Numerical demonstrations\label{ssec:remez-numerical-dems}}
In this subsection, we numerically test \cref{alg:constrained_remez_qsp} on QSP problems of interest, that we introduced in \Cref{ssec:expt}, with the parameters set to be the same as mentioned there. 
\begin{figure}
    \centering
    \begin{subfigure}[t]{0.43\linewidth}
        \centering
        \includegraphics[width=\linewidth]{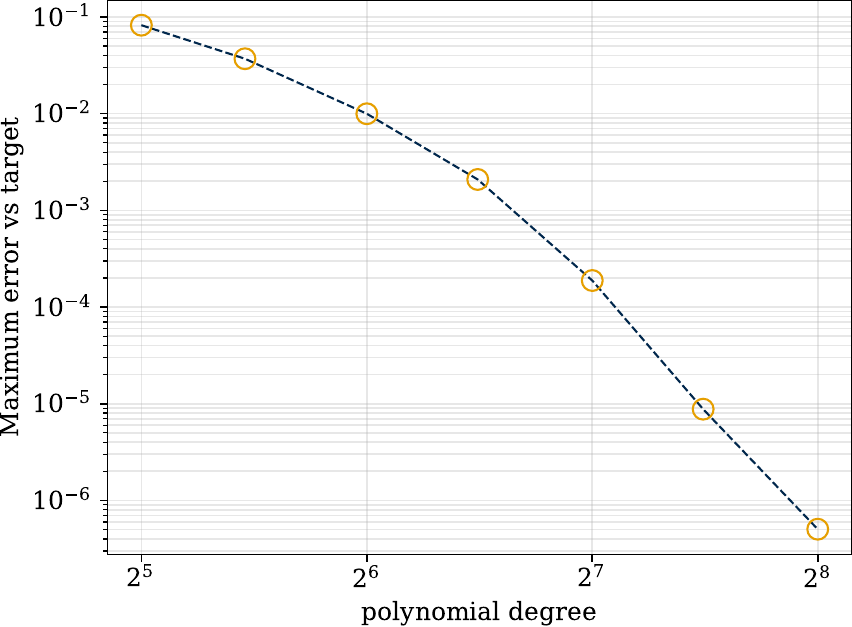}
        \caption{Error scaling.}
        \label{fig:Remez_threshold_projection}
    \end{subfigure}
    \hfill
    \begin{subfigure}[t]{0.55\linewidth}
        \centering
        \includegraphics[width=\linewidth]{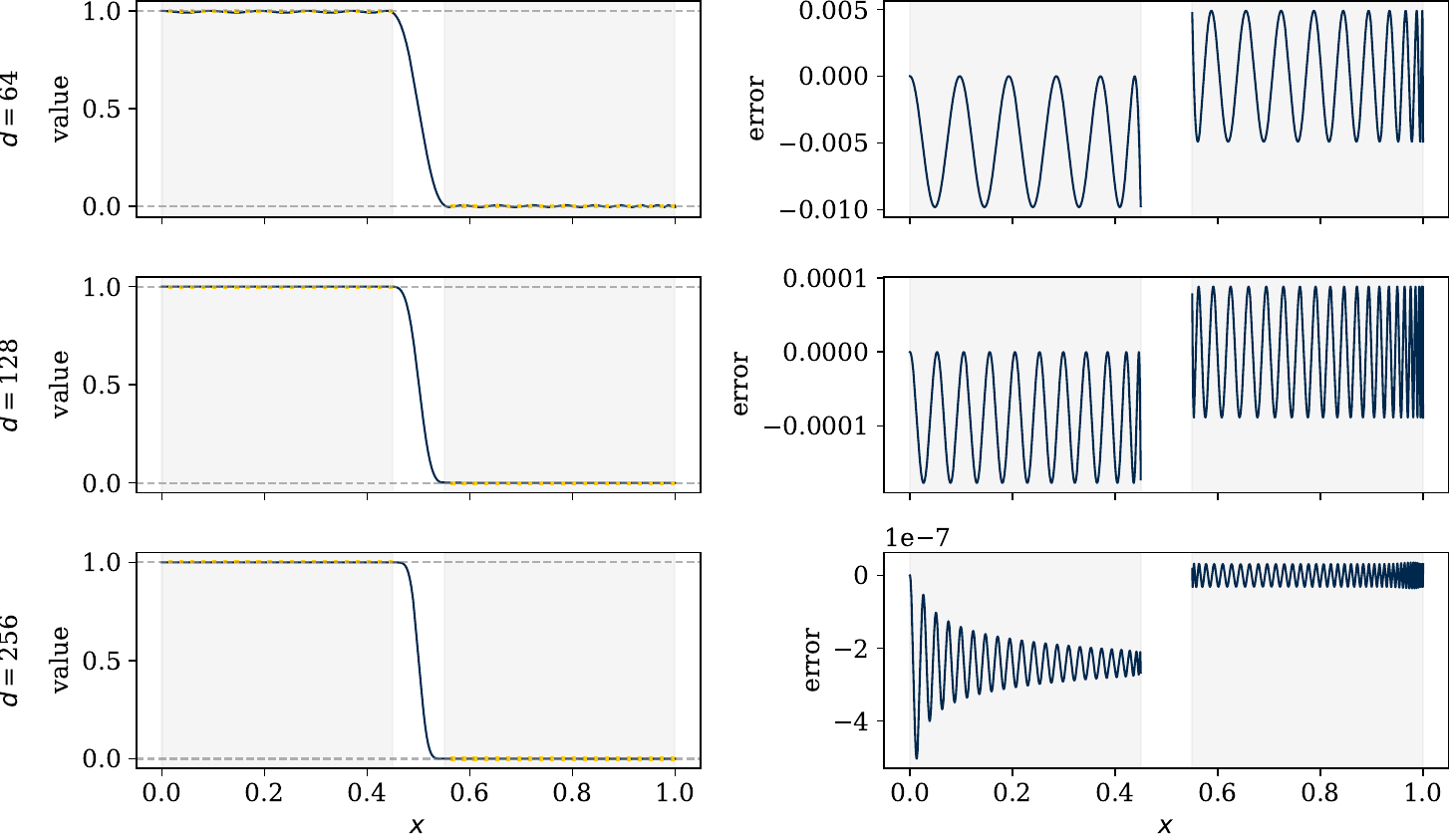}
        \caption{Pointwise approximation error. Maize dashed lines stand for the target function $g_{\strut \text{TP}}$.}
        \label{fig:Remez_threshold_projection_scaled_error}
    \end{subfigure}
    
    \vspace{-15pt}
    \caption{Performance of the Remez and active-set based method for $g_{\strut \text{TP}}$.}
    \label{fig:remez_thresh_proj}
    \vspace{-20pt}
\end{figure}

We first test the algorithm on the threshold projector target function $g_{\strut \text{TP}}$. The results are depicted in \Cref{fig:remez_thresh_proj}. Note that the error agrees well with the results in \Cref{fig:degree_scaling_thresh_proj} attained from methods described in \Cref{sec:nonlinear-retraction}; the only suboptimal result occurs at $d = 256$. From the pointwise error plot, we find that other degree values demonstrate an equioscillating pattern, suggesting an optimal result by the Remez-like construction. However, when $d = 256$, the result lacks equioscillation, justifying the suboptimality of the approximation. Though other numerical-trick parameters were explored, the result did not significantly improve.

\begin{figure}
    \centering
    \begin{subfigure}[t]{0.43\linewidth}
        \centering
        \includegraphics[width=\linewidth]{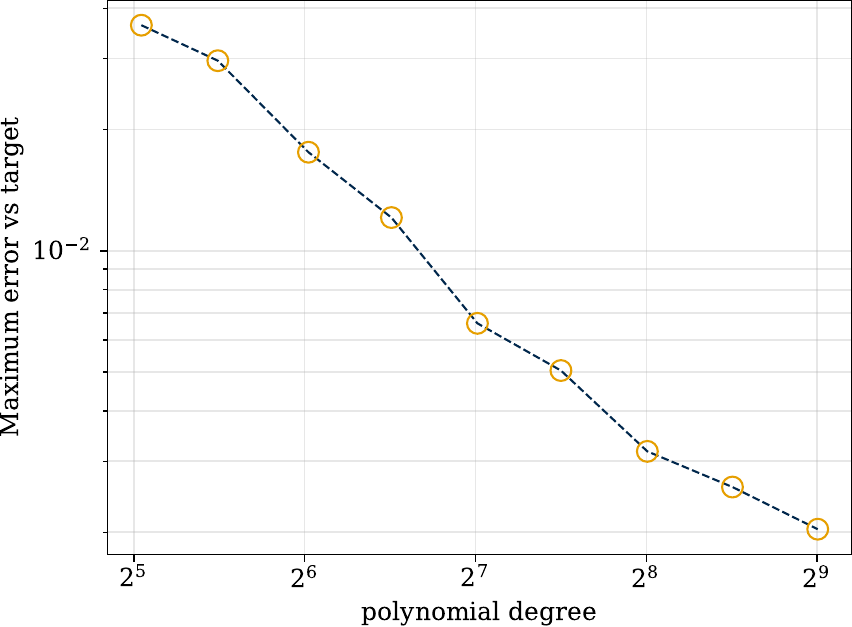}
        \caption{Error scaling.}
        \label{fig:Remez_linear_system}
    \end{subfigure}
    \hfill
    \begin{subfigure}[t]{0.55\linewidth}
        \centering
        \includegraphics[width=\linewidth]{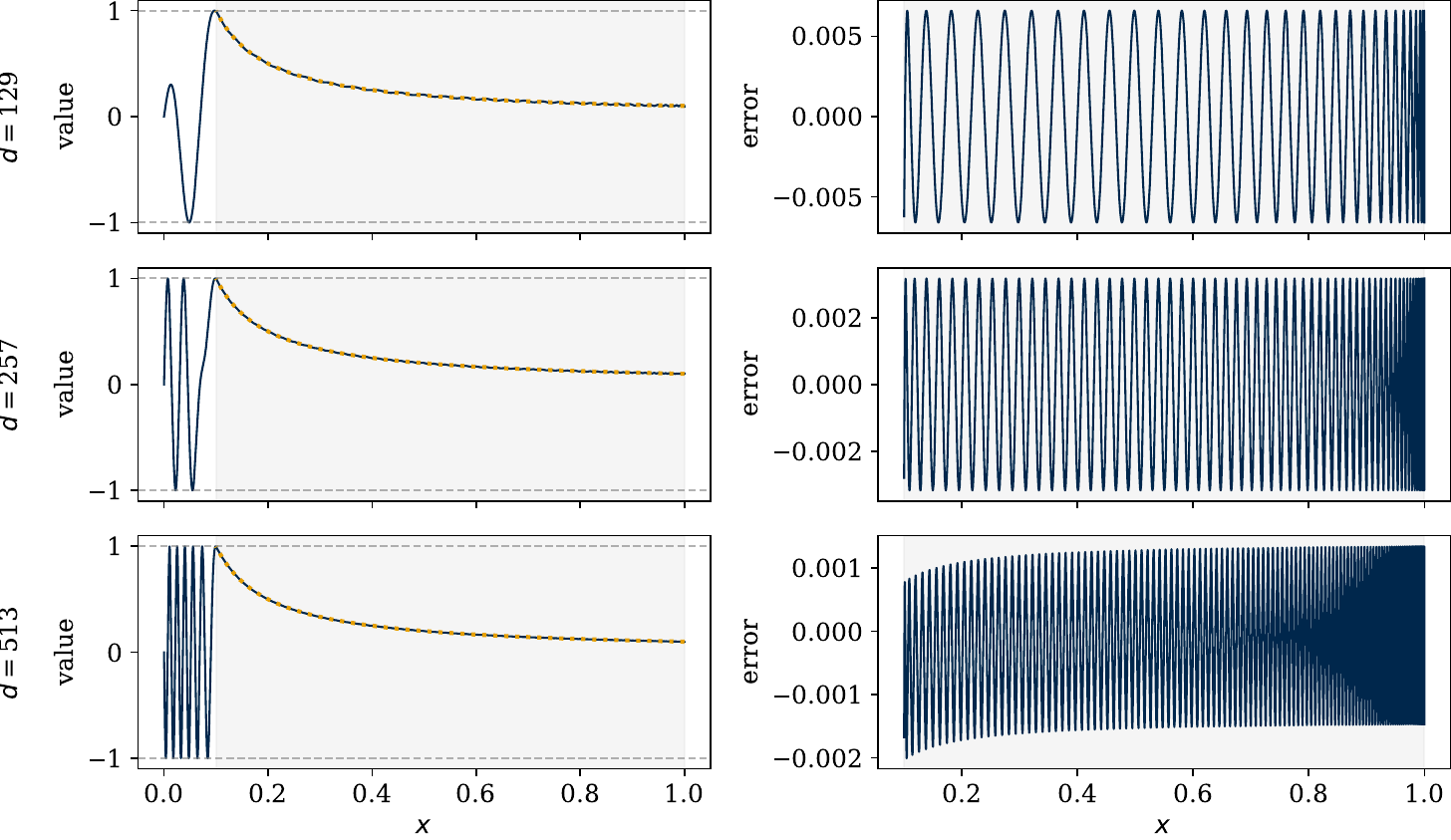}
        \caption{Pointwise approximation error. Maize dashed lines stand for the target function $g_{\strut \text{MI}}$.}
        \label{fig:Remez_linear_system_error}
    \end{subfigure}
    
    \vspace{-15pt}
    \caption{Performance of the Remez and active-set based method for $g_{\strut \text{MI}}$.}
    \label{fig:remez_linear_system}
    \vspace{-30pt}
\end{figure}

We then test the matrix inversion problem. From \cref{fig:remez_linear_system}, we again see that the performance agrees well with the results in \cref{fig:degree_scaling_mat_inv}  for $d \leq 257$. We further visualize the pointwise error in \cref{fig:Remez_linear_system_error}. The plot shows that the equioscillating structure breaks down near the endpoint for $d = 513$. This breakdown is due to the approximation polynomial requiring a small buffer region to transition from the upper bound value one to smaller values, which resembles a reflection effect near the boundary. This leads to suboptimal behavior near the endpoint. 

\begin{figure}
    \centering
    \begin{subfigure}[t]{0.43\linewidth}
        \centering
        \includegraphics[width=\linewidth]{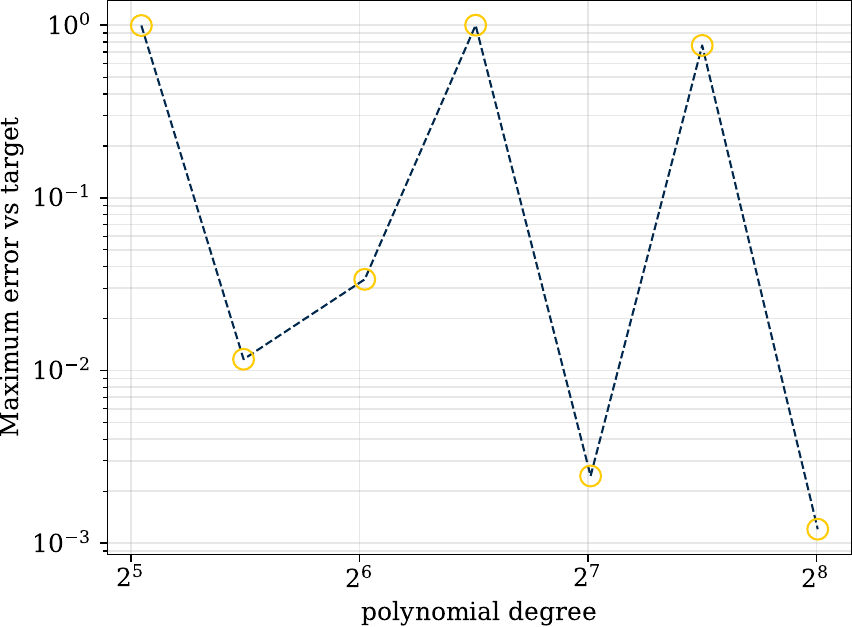}
        \caption{Error scaling.}
        \label{fig:Remez_amplitude_amplification}
    \end{subfigure}
    \hfill
    \begin{subfigure}[t]{0.55\linewidth}
        \centering
        \includegraphics[width=\linewidth]{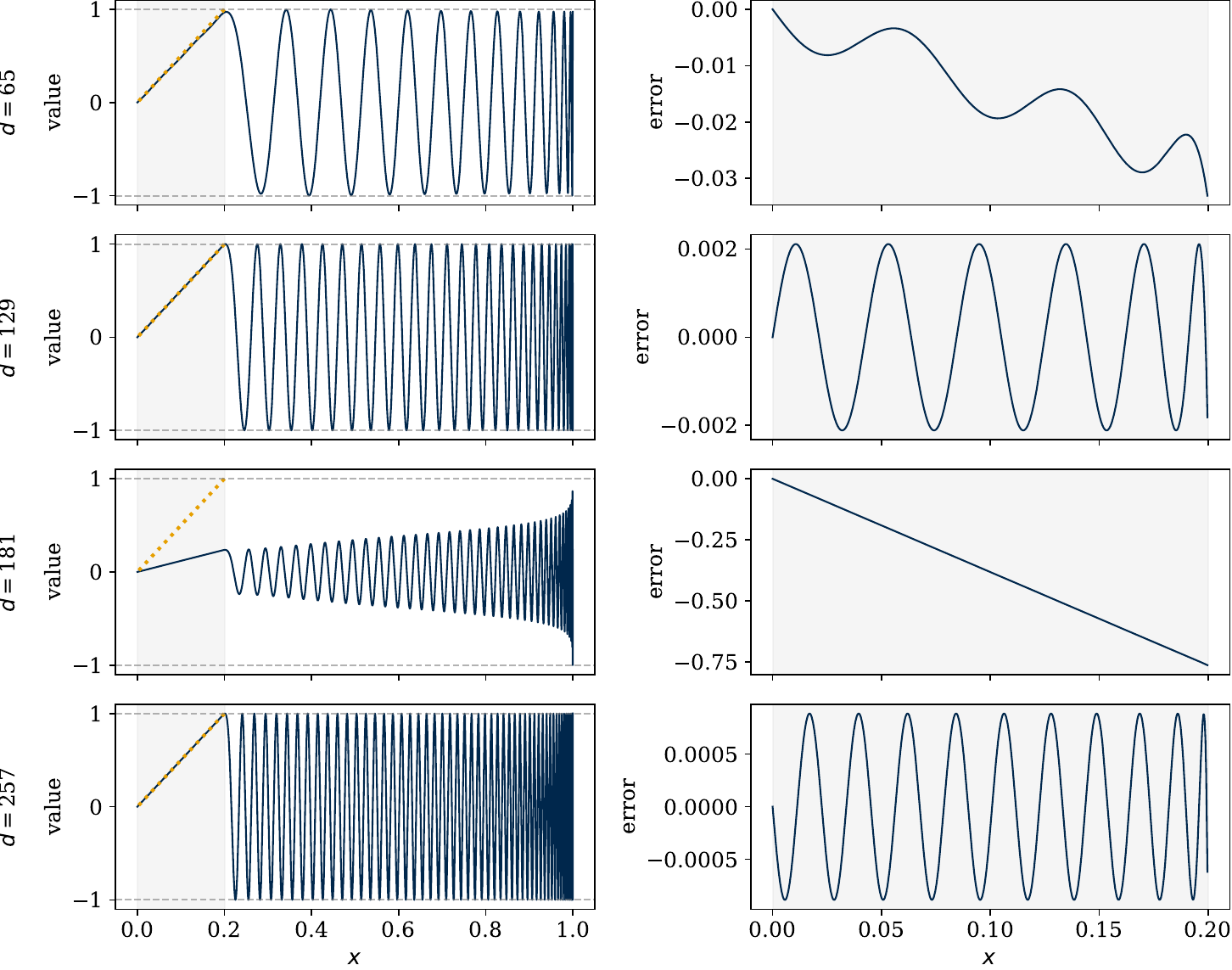}
        \caption{Pointwise approximation error. Maize dashed lines stand for the target function $g_{\strut \text{SV}}$.}
        \label{fig:Remez_amplitude_amplification_error}
    \end{subfigure}
    
    \vspace{-15pt}
    \caption{Performance of the Remez and active-set based method for $g_{\strut \text{SV}}$.}
    \label{fig:remez_amplitude_amplification}
    \vspace{-30pt}
\end{figure}

Finally, we test the case of uniform singular value amplification. In practice, this function is often approximated by taking the coefficients for even polynomials approximating rectangular functions and increasing the degree by one. However, these approximations will not be optimal because they require values of $x$ outside $[0,\Gamma]$ to approximate zero. The numerical results in \cref{fig:remez_amplitude_amplification} show more complicated behavior than previous numerical experiments. In particular, the approximation error does not necessarily decrease as $d$ increases. From the pointwise error plots, we observe that the equioscillating pattern is not always preserved, and the numerical behavior can vary significantly with the degree $d$. 

A possible explanation comes from the structure of the target function. Since the target is proportional to $x$, it is itself one of the basis functions in the construction. On the target interval $Y$, the algorithm may therefore incorrectly interpret this low-degree component as already providing a satisfactory approximation. However, outside $Y$, higher-degree components are still needed to enforce boundedness and produce the required nontrivial behavior. This tension can lead either to a breakdown of the equioscillating pattern or to a large overshoot error. In either case, the resulting approximation quality becomes poor. This phenomenon is not observed in the previous two examples because their target functions have nontrivial projections onto a broader range of basis components, so the algorithm is less likely to become trapped in such a misleading low-degree approximation.

To summarize, the Remez and active-set based algorithm in \cref{alg:constrained_remez_qsp} combines a Remez-like mechanism for maintaining an equioscillating approximation pattern on the target set with an active-set strategy for enforcing boundedness constraints. This makes the method computationally efficient and reduces the need for heavy discretization over the full domain. At the same time, our numerical experiments show that the method may become suboptimal or even highly unstable in certain cases, especially when the active constraints interact unfavorably with the equioscillation structure. These observations highlight the limitations of the present iterative construction and motivate the development below of a more versatile and robust procedure for polishing approximation polynomials.

\section{Nonlinear Fourier retraction}
\label{sec:nonlinear-retraction}
\label{ssec:retraction-intro}

In this section, we develop a heuristic algorithm to overcome the challenge of nearly feasible solutions of \cref{eq:intro-opt-prob} in the QSP workflow, and then apply the resulting algorithm to the QSP problems from \Cref{ssec:expt}.

\subsection{The algorithm}
\label{ssec:alg-retraction}
The Weiss algorithm, as introduced in \cite[Section~2.3]{Alexis2025} and as applicable to our setting, takes as input a polynomial $b(z)$ with $|b(z)| \neq 1$ almost everywhere (a.e.)\footnote{This implies $\log (1 - |b(z)|^2) \in L^1(\TT)$, which is called the Szeg\H{o} condition in \cite{Alexis2025}.} on $\TT$, and $\norm{b}{L^\infty(\TT)} \le 1$ (this ensures $|b(z)| < 1$ for all $z \in \DD$ by the maximum modulus principle), and outputs the unique outer polynomial $a^\ast(z)$ such that $a(z)a^\ast(z) + b(z)b^\ast(z) = 1$ identically for all $z \in \CC$. The algorithm is defined by the following relations, and the computations run left to right (here we have modified the algorithm, as compared to \cite{Alexis2025}, to directly output $a^\ast$ instead of $a$):
\vspace*{-0.2cm}
\begin{equation}
\label{eq:weiss}
R(z) = \frac{1}{2}\log (1-|b(z)|^2), \quad G(z) = R(z) + i \mathcal{H}(R(z)), \quad a^\ast(z) = e^{G(z)}, \quad z \in \TT,
\vspace*{-0.2cm}
\end{equation}
where $\mathcal{H}(\cdot)$ denotes the Hilbert transform\footnote{The Hilbert transform is defined for a smooth function $f$ on $\TT$ by the principal value formula $\mathcal{H}(f)(e^{i\zeta}) := (2 \pi)^{-1} \text{p.v.} \int_{0}^{2 \pi} f(e^{i \theta}) \cot \left( \frac{\zeta - \theta}{2} \right) d\theta$.}. For example, its action on trigonometric polynomials is given by $\mathcal{H}(e^{ik\theta}) := -i \;\text{sgn}(k) e^{i k \theta}, \; \theta \in [0,2\pi], \; k \in \mathbb{Z}$, where $\text{sgn}(k) = -1, 0, 1$ for $k < 0$, $k=0$ and $k>0$ respectively. Under the assumptions on $b$, one can show that $G$ is holomorphic on $\DD$, $a^\ast(0)$ is real and positive (see the proof of \cite[Theorem~4]{Alexis2025} for details), and if all coefficients of $b$ are real, then the same holds for the coefficients of $a^\ast$ in the monomial basis.

We now return to our modification of the Weiss algorithm, that we call \textit{modified Weiss}, or \textsc{MWeiss} in short, that handles the situation when $\norm{b}{L^\infty(\TT)} > 1$ (still with $|b(z)| \neq 1$ a.e. on $\TT$), which is derived from the case of a nearly feasible polynomial $P(x)$, as we explained before in \Cref{ssec:qsp-nlft}. The basic algorithm is still given by \cref{eq:weiss}, and the main point of difference of this new algorithm is the definition of $R(z)$, when $|b(z)| > 1$, in which case the argument of $\log$ becomes negative. Let us split $\TT$ into disjoint sets $\TT_{+} := \{z \in \TT: |b(z)| > 1\}$,  $\TT_{-} := \{z \in \TT: |b(z)| < 1\}$, and $\TT_{0} := \{z \in \TT: |b(z)| = 1\}$. Since $b$ is a polynomial and $|b(z)| \ne 1$ a.e. on $\TT$, we know that $\TT_0$ is finite, but both $\TT_{+}$ and $\TT_{-}$ can have positive Haar measure. On $\TT_{-}$, the definition of $R$ stays the same as before, while on $\TT_{+}$, we use the principal branch of $\log$ to define $R$. This leads to the following modified definition:
\vspace*{-0.2cm}
\begin{equation}
\label{eq:R-def}
R(z) = \frac{1}{2}\log (1-|b(z)|^2) := 
\begin{cases}
    \frac{1}{2}\log \left| 1-|b(z)|^2 \right|, \quad & z \in \TT_{-} \\
    \frac{1}{2}\log \left| 1-|b(z)|^2 \right| + \frac{i \pi}{2}, \quad & z \in \TT_{+}
\end{cases}
\vspace*{-0.1cm}
\end{equation}
and note that since $\log \left| 1-|b(\cdot)|^2 \right| \in L^p(\TT)$, for all $1 \le p < \infty$, the same is also true for $R$. Since the Hilbert transform is a bounded operator on $L^p(\TT)$, for every $1 < p < \infty$, this now implies that $\mathcal{H}(R), G \in L^p(\TT)$, for all $1 \le p < \infty$. The Carleson-Hunt theorem (see e.g. \cite{fefferman1973pointwise}) says the Fourier series of $R, \mathcal{H}(R)$, and $G$ converge point-wise a.e. on $\TT$. Let the Fourier series expansion of $R$ be given by $R(e^{i \theta}) = \sum_{k \in \mathbb{Z}} \hat{R}_k e^{i k \theta}$. Then the Hilbert transform yields the Fourier series expansions of $\mathcal{H}(R)$ and $G$:
\vspace*{-0.3cm}
\begin{equation}
\label{eq:HR-G-FS}
\mathcal{H}(R)(e^{i \theta}) = -i \sum_{k \ge 1} (\hat{R}_k e^{i k \theta} - \hat{R}_{-k} e^{-i k \theta}), \quad G(e^{i \theta}) = \hat{R}_0 + 2 \sum_{k \ge 1} \hat{R}_k e^{ik\theta}.
\vspace*{-0.3cm}
\end{equation}
Since $G$ has no negative Fourier modes, we may extend $G$ to a holomorphic function in $\DD$, defined by its Fourier series. Then it is clear that $a^\ast$, being the exponential of a holomorphic function, is also holomorphic with no zeros in $\DD$. However, we no longer have that $a^\ast(0)$ is real and positive (see properties of NLFT in \cite{Alexis2024,ni2025inversenonlinearfastfourier}); instead, we have the following theorem, whose proof can be found in \Cref{appssec:proof-main-theorem}:
\begin{theorem}
\label{thm:bounds-astar-real-imag}
Let $b$ be a polynomial such that $\norm{b}{L^{\infty}(\TT)} > 1$, $|b(z)| \ne 1$ a.e. on $\TT$, and let $a^\ast$ be computed as in the \textsc{MWeiss} algorithm. Then $a^\ast$ is holomorphic with no zeros in $\DD$, and its real and imaginary parts satisfy 
\vspace*{-0.2cm}
\begin{equation}
\label{eq:astar-formulas}
    \Re(a^\ast(0)) = K \cos \left( \frac{\pi |\TT_{+}|}{2} \right), \quad \Im(a^\ast(0)) = K \sin \left(\frac{\pi |\TT_{+}|}{2} \right),
\vspace*{-0.2cm}
\end{equation}
where $K = e^{\frac{1}{4 \pi} \int_{0}^{2 \pi} \log \left| 1 - |b(z)|^2 \right| \; d\theta} > 0$, $\TT_{+}$ is as defined above, and $|\TT_{+}|$ denotes the Haar measure of $\TT_{+}$ (the measure is normalized so that $\TT$ has measure $1$). Moreover, $a^\ast \in L^p(\TT)$, for every $1 \le p \le \infty$.
\end{theorem}

The lemma shows that unless $|\TT_{+}| = 1$, we always have $\Re(a^\ast(0)) > 0$, and if $|\TT_{+}|$ is small, then $\Im(a^\ast(0))$ is small, scaling as $\Im(a^\ast(0) \sim K \pi |\TT_{+}|/2$. In general, for the other Fourier series coefficients of $a^\ast$, their imaginary parts are also small, if $|\TT_{+}|$ is small. Moreover, $a^\ast \in L^2(\TT)$ implies that (i) the Fourier series coefficients are square summable (by Plancherel's theorem), and in particular, they decay, and (ii) the Fourier series of $a^\ast$ converges point-wise a.e. on $\TT$. In the final step of $\textsc{MWeiss}$, we set to zero the Fourier series coefficients of the $k^{\text{th}}$ Fourier mode of $a^\ast$, for every $k > d$, and discard the imaginary parts of the remaining Fourier series coefficients --- i.e. if $a^\ast(e^{i\theta}) = \sum_{k\ge 0} \hat{a}^\ast_k e^{ik\theta}$, then the output of \textsc{MWeiss} are the Fourier series coefficients of $\tilde{a}^\ast(e^{i\theta}) = \sum_{k=0}^{d} \Re(\hat{a}^\ast_k) e^{ik\theta}$.

The pseudocode of \textsc{MWeiss} is presented in \Cref{alg:modweiss}. In the algorithm, we choose a discretization parameter $N$, whose value is chosen large enough, so that the Fourier series coefficients are computed with high accuracy, and the truncation and aliasing errors due to using a finite $N$ are small (this is theoretically justified as the quantities $R, G, a^\ast \in L^2(\TT)$). Our fast Fourier transform (FFT) and inverse fast Fourier transform (IFFT) are consistent with the \texttt{NumPy} conventions \cite{schatzman1996accuracy,harris2020array}; i.e. for $\mathbf{x} = (x_0,\dots,x_{N-1}) \in \CC^N$, we define $\text{FFT}(\mathbf{x})_k := \sum_{j=0}^{N-1} x_j e^{-\frac{i 2\pi jk}{N}}$, and $\text{IFFT}(\mathbf{x})_k := \frac{1}{N}\sum_{j=0}^{N-1} x_j e^{\frac{i 2\pi jk}{N}}$. The bold vectors in \Cref{alg:modweiss} are elements of $\CC^N$ (except $\tilde{\mathbf{a}}^\ast \in \mathbb{R}^{d+1}$), and note that $\mathbf{z}$ is never formed; instead one directly calculates $b(\mathbf{z}) \in \CC^N$ sampled at the points $\mathbf{z}$, using IFFT in Line~1. The $\frac{1}{N}$ factors in Lines~2, 4 are needed due to the \texttt{NumPy} FFT conventions. The quantities $\mathbf{R}$ and $\mathbf{G}$ are the values of $R$ and $G$ respectively, sampled at $\mathbf{z}$; while $\hat{\mathbf{R}}$ and $\hat{\mathbf{a}}^\ast$ are the Fourier series coefficients of $R$ and $a^\ast$ respectively, over the Fourier mode indices $k=-\lfloor N/2 \rfloor,\dots, N-1-\lfloor N/2 \rfloor$. In Line~5, the subscript notation $\Re(\hat{\mathbf{a}}^*)_{0:d}$ means that we are only selecting the Fourier mode coefficients with indices $k=0,\dots,d$.

\vspace*{-0.2cm}
\begin{algorithm}
\caption{\textsc{MWeiss}}
\label{alg:modweiss}
    \begin{algorithmic}[1]
    \REQUIRE Monomial coefficients of polynomial $b$ of degree $d$, Integer $N$ 
    \ENSURE Monomial coefficients of polynomial $\tilde{a}$
    \STATE $\mathbf{R} = \log \sqrt{1 - |b(\mathbf{z})|^2}$, $\mathbf{z} = \left\{ e^{i 2\pi j / N} \right\}_{j=0}^{N-1}$ \hfill \COMMENT{Evaluation of $b(\mathbf{z})$ via IFFT}
    \STATE $\hat{\mathbf{R}} = \frac{1}{N}\text{FFT}(\mathbf{R})$, with entries indexed $\hat{R}_{-\lfloor N/2 \rfloor}, \dots, \hat{R}_{N-1-\lfloor N/2 \rfloor}$
    \STATE $\mathbf{G} = \hat{R}_0 + 2 \sum_{k = 1}^{N-1-\lfloor N/2 \rfloor} \hat{R}_{k} \mathbf{z}^k$ \hfill \COMMENT{Evaluation via IFFT}
    \STATE $\hat{\mathbf{a}}^* = \frac{1}{N}\text{FFT}\left( e^{\mathbf{G}} \right)$
    \RETURN $\tilde{\mathbf{a}}^\ast=\Re(\hat{\mathbf{a}}^*)_{0:d}$
    \end{algorithmic}
\end{algorithm}
\vspace*{-0.2cm}
\vspace*{-0.2cm}
\begin{algorithm}
\caption{\textsc{Nonlinear Fourier Retraction}}\label{alg:retraction}
    \begin{algorithmic}[1]
    \REQUIRE Chebyshev coefficients $\mathbf{c}$ of degree $d$ polynomial $P \in \mathbb{R}[x]$, Integer $N$
    \ENSURE QSP phase factors $\Psi$, Chebyshev coefficients of $\mathcal{R}(P)$
    \STATE $\mathbf{b} =$ Monomial coefficients of $b \in \CC[z]$ \hfill \COMMENT{Using \Cref{appssec:coefmap} based on $P$}
    \STATE $\tilde{\mathbf{a}}^* = \textsc{MWeiss}(\mathbf{b}, N)$ \hfill \COMMENT{\Cref{alg:modweiss}}
    \STATE $\boldsymbol{\gamma} = \text{INFFT}(\tilde{\mathbf{a}}^*, \mathbf{b})$ \hfill \COMMENT{Algorithm~1 of \cite{ni2025inversenonlinearfastfourier}}
    \STATE $\Psi = \arctan (\boldsymbol{\gamma})$ \hfill \COMMENT{QSP-NLFT correspondence}
    \STATE $\left( \begin{smallmatrix} a_1(z) & b_1(z) \\ -b_1^*(z) & a_1^*(z) \end{smallmatrix} \right) = \overbrace{\boldsymbol{\gamma}}(z)$
    \STATE $\mathcal{R}(\bc) = $ Chebyshev coefficients of $\mathcal{R}(P)$ \hfill \COMMENT{$\mathcal{R}(P)(\cos\theta) := \Re(b_1(e^{2i \theta}) e^{-i d \theta})$}
    \RETURN $\Psi$, $\mathcal{R}(\bc)$
    \end{algorithmic}
\end{algorithm}
\vspace*{-0.2cm}

Assuming $b$ has real coefficients and $|\TT_{+}| < 1$, after \textsc{MWeiss}, we have two polynomials $b, \tilde{a}^\ast$ of degree $d$, both with real coefficients, and crucially $\tilde{a}^\ast(0) > 0$. Now, it is argued in \cite[Section~4.1]{ni2025inversenonlinearfastfourier} that the layer stripping algorithm is well-defined on any such pair of polynomials (this only uses the condition $\tilde{a}^\ast(0) > 0$), and the output of layer stripping is some $\bga \in \CC^{d+1}$. The extra condition of real valued  coefficients of $b$ and $\tilde{a}^\ast$ strengthens this conclusion to $\bga \in \mathbb{R}^{d+1}$, which then allows us to use the QSP-NLFT correspondence $\gamma_k = \tan \psi_k$, for all $0 \le k \le d$, to extract a set of QSP phase angles $\Psi$. This is one of the main reasons why we discarded the imaginary parts of all the Fourier series coefficients $\hat{a}_k^\ast$, and not just of $\hat{a}_0^\ast$, in the \textsc{MWeiss} algorithm. Finally, one computes the NLFT of $\bga$ to get $\overbrace{\boldsymbol{\gamma}}(z) := \left( \begin{smallmatrix} a_1(z) & b_1(z) \\ -b_1^*(z) & a_1^*(z) \end{smallmatrix} \right) \in \mathrm{SU}(2)$; thus the polynomials $b_1$ and $a_1^\ast$ satisfy $aa^* + bb^*=1$ on $\CC$. If we define $\mathcal{R}(P) \in \mathbb{R}[x]$ by $\mathcal{R}(P)(\cos \theta)=\Re(b_1(e^{2i\theta})e^{-id\theta})$, then the QSP-NLFT correspondence \cref{eq:HUdH-lemma-nonsym} says that $\norm{\mathcal{R}(P)}{X} \le 1$. We call this polynomial $\mathcal{R}(P)$ the \textit{nonlinear retraction} of $P$. Note that all the facts stated in this paragraph also hold when layer stripping is replaced by the INFFT algorithm. Combining both the \textsc{MWeiss} and the subsequent layer stripping (or INFFT) steps, we get the full algorithm called \textit{nonlinear Fourier retraction} (or simply, \textit{retraction}), and a pseudocode is presented in \Cref{alg:retraction}. The output of the algorithm are $\bga$ and the Chebyshev coefficients $\mathcal{R}(\bc)$ of the polynomial $\mathcal{R}(P)$, where $\bc$ denotes the Chebyshev coefficients of the input polynomial $P$. As compared to $P$, the nonlinear retraction $\mathcal{R}(P)$ is a feasible polynomial.

\vspace{0.1cm}
\noindent
\textit{Remarks on stability}: In Line~1 of \textsc{MWeiss}, it is possible that the computation of $\log$ becomes unstable when $1 - |b(z)|^2$ is close to zero, at some subset of the points $\mathbf{z}$. In this case, we propose rotating the grid slightly, and instead sample the function $|b(z)|$ at $\mathbf{z}' =  \left\{ e^{i (\delta \theta +2\pi j / N)} \right\}_{j=0}^{N-1}$, for some suitably chosen $\delta \theta$, so that $|b(z)| \neq 1$ at any point in $\mathbf{z}'$, and then update Line~2 accordingly for calculating the Fourier series coefficients. In practice, we did not need to perform this safeguarding step for the numerical examples that we demonstrate next in \Cref{ssec:numerical-examples} and \Cref{app:retraction}. Another potential source of numerical instability is the INFFT (or layer stripping) step in Line~3 of \Cref{alg:retraction}, since the sufficient conditions guaranteeing their numerical stability \cite{ni2025inversenonlinearfastfourier} no longer holds in our setting. In our numerical examples, we did not encounter this numerical instability either, but if this does happen, we propose using extended-precision arithmetic for the INFFT (or layer stripping) computation.

\subsection{Numerical demonstrations}
\label{ssec:numerical-examples}
We now apply \Cref{alg:retraction} to the QSP problems from \Cref{ssec:expt}. The uniform singular value amplification problem is presented in full in \Cref{subsubsec:unif_amp}. An important numerical finding for the matrix inversion problem is discussed in \Cref{subsubsubsec:mat_inv}, and extra supporting results are delegated to \Cref{appssec:retract_exp_matrix}. Similar results for the threshold projector problem are presented in \Cref{appssec:retract_exp_threshold}. For these demonstrations, $N$ will always refer to the discretization parameter used in \textsc{MWeiss}.

We compare several aspects of the nonlinear Fourier retraction algorithm in these examples, e.g., we show how the supremum norm on the fitting set $Y$ of the difference between the target function and the polynomial corresponding to the optimal solution of the LSIP \cref{eq:lsip-prob} scales with the maximum polynomial degree $d$, and then demonstrate how the results improve when nonlinear retraction is applied to the polynomial instead. In these examples, we frequently refer to a `\texttt{npts}' parameter, which controls how the sets $\hatX$ and $\hatY$ are chosen for the LSIP problem \cref{eq:lsip-prob}. Specifically, $\hatX$ consists of those Chebyshev nodes of the second kind ($\cos(k\pi/\texttt{npts}), \; 0 \leq k \leq \texttt{npts}$), that lie in $X$, and then define $\hatY := \hatX \cap Y \subseteq \hatX$. In all instances, unless otherwise stated, the resulting LSIP was solved using \texttt{cvxpy}, with the solver choice being \texttt{SCS} (or \texttt{CLARABEL} as a replacement if \texttt{SCS} fails to converge, determined on a case-by-case basis). We also demonstrate the performance of the method as \texttt{npts} increases.

\subsubsection{Performance for uniform singular value amplification}
\label{subsubsec:unif_amp}
Recall that for this example with the target function $g_{\strut \text{SV}}$, the Remez- and active-set-based method was numerically unstable. In the first illustration (\Cref{fig:retraction_ex}), we use $d=101$, and $\texttt{npts}=128$. The optimal solution to \cref{eq:lsip-prob} and the corresponding polynomial are denoted $\hbc^\ast$ and $\inner{\hbc^\ast}$ respectively. In the left panel, we plot $g_{\strut \text{SV}}$, $\inner{\hbc^\ast}$ and its nonlinear retraction $\mathcal{R}(\inner{\hbc^\ast})$. In the right panel, we plot over $Y$, the functions $|\inner{\hbc^\ast}(x) - g_{\strut \text{SV}}(x)|$, $|\mathcal{R}(\inner{\hbc^\ast})(x) - g_{\strut \text{SV}}(x)|$, and $|\texttt{Scaled}(\inner{\hbc^\ast})(x) - g_{\strut \text{SV}}(x)|$, where $\texttt{Scaled}(\inner{\hbc^\ast})$ is the rescaling of the polynomial $\inner{\hbc^\ast}$ by the maximum constraint violation on $X$ to enforce feasibility. We note that the retracted polynomial demonstrates smaller maximum absolute error with respect to the target function than the rescaling procedure, over the fitting set $Y$. This rescaling procedure is currently the standard approach to handle constraint violations in the fully-coherent regime \cite{Gily_n_2019}; alternatively, for the uniform singular value amplification problem, \cite{Martyn2021} suggests to multiply the infeasible polynomial by a window function, but this results in an additive factor to the overall degree. Our results show that retraction is a better alternative for this problem.
\begin{figure}[htp!]
\vspace*{-5pt}
    \centering
    \includegraphics[width=\linewidth,height=0.25\textheight]{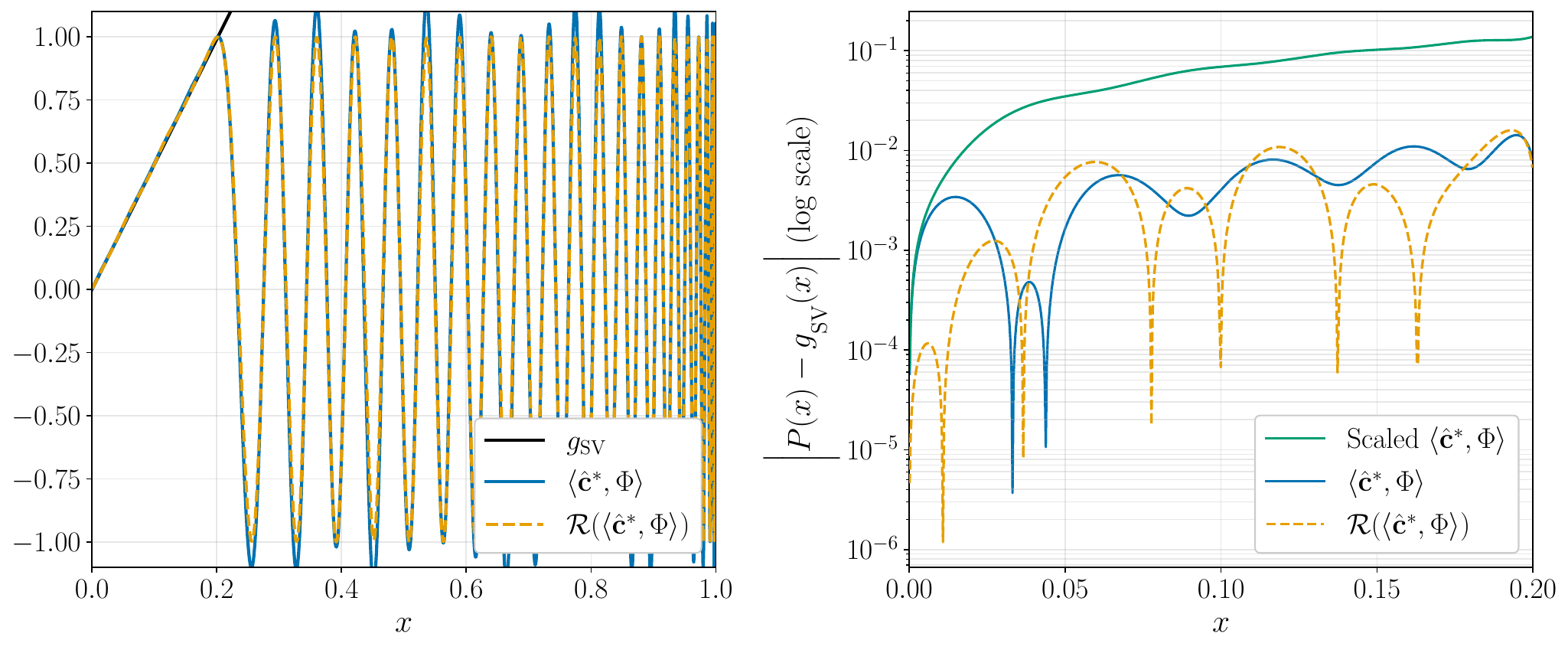}
    \vspace{-24pt}
    \caption{Illustration of nonlinear Fourier retraction for the uniform singular value amplification example (target function $g_{\strut \text{SV}}$). The polynomial $\inner{\hbc^\ast}$ is the solution to LSIP \cref{eq:lsip-prob} with polynomial degree $d=101$, and $\texttt{npts} = 128$. The (left) panel plots $g_{\strut \text{SV}}$ in black, $\inner{\hbc^\ast}$ in blue, and the retraction result $\mathcal{R}(\inner{\hbc^\ast})$ in dashed maize. In the (right) panel, we plot the absolute difference $|P(x) - g_{\strut \text{SV}}(x)|$ over $Y$, for the cases $P$ is $\inner{\hbc^\ast}$, the retraction result $\mathcal{R}(\inner{\hbc^\ast})$, and in green the polynomial obtained by scaling $\inner{\hbc^\ast}$, as mentioned in the text.}
    \label{fig:retraction_ex}
    \vspace{-18pt}
\end{figure}

In the left panel of \Cref{fig:polyspace_conv_unif_amp}, we demonstrate the performance of retraction as we increase \texttt{npts} for the LSIP \cref{eq:lsip-prob}, and use $d=101$ throughout. The accuracy of the approximating polynomials, either of the optimal solution of \cref{eq:lsip-prob} or its nonlinear retraction, is measured using the $2$-norm of the difference of the Chebyshev coefficients of these polynomials and $\hbc^\ast_{\text{gr}}$. The Chebyshev coefficients $\hbc^\ast_{\text{gr}}$ correspond to the optimal solution of \cref{eq:lsip-prob}, computed by using a very high value of \texttt{npts}$=2^{19}$, and acts as an accurate approximation of the optimal solution of \cref{eq:intro-opt-prob}, even though this polynomial exhibits a constraint violation on the order of $10^{-8}$ (as already demonstrated in \cref{fig:constraintviol}). We note from the plot that $\mathcal{R}(\inner{\hbc^\ast})$ achieves a lower error  for values of \texttt{npts} from $2^7$ to $2^{17}$, as compared to $\inner{\hbc^\ast}$. We also indicate on this plot whether the polynomials are feasible or infeasible. The right panel of \Cref{fig:polyspace_conv_unif_amp} shows the same results for the target function  $(1-\epsilon) g_{\strut \text{SV}}$, with $\epsilon=10^{-4}$, along with an additional important distinction: when solving \cref{eq:lsip-prob}, the second set of constraints are tightened to $-(1 - \epsilon) \le \langle \bc , \Phi(x) \rangle \le 1 - \epsilon$, for all $x \in X$. This represents a slightly relaxed setting as compared to the fully-coherent regime of left panel. Unlike in the left panel, we see that $\inner{\hbc^\ast}$ is a feasible polynomial for $\texttt{npts} \ge 2^{13}$, and thus we have $\mathcal{R}(\inner{\hbc^\ast})=\inner{\hbc^\ast}$ on this range. We also see that $\mathcal{R}(\inner{\hbc^\ast})$ achieves an equal or lower error than $\inner{\hbc^\ast}$, for all values of \texttt{npts}. For the left and right panels in this demonstration, we used $N=2^{18}$ and $N=2^{16}$ respectively.
\begin{figure}[htp!]
    \centering
    \includegraphics[width=0.8\linewidth]{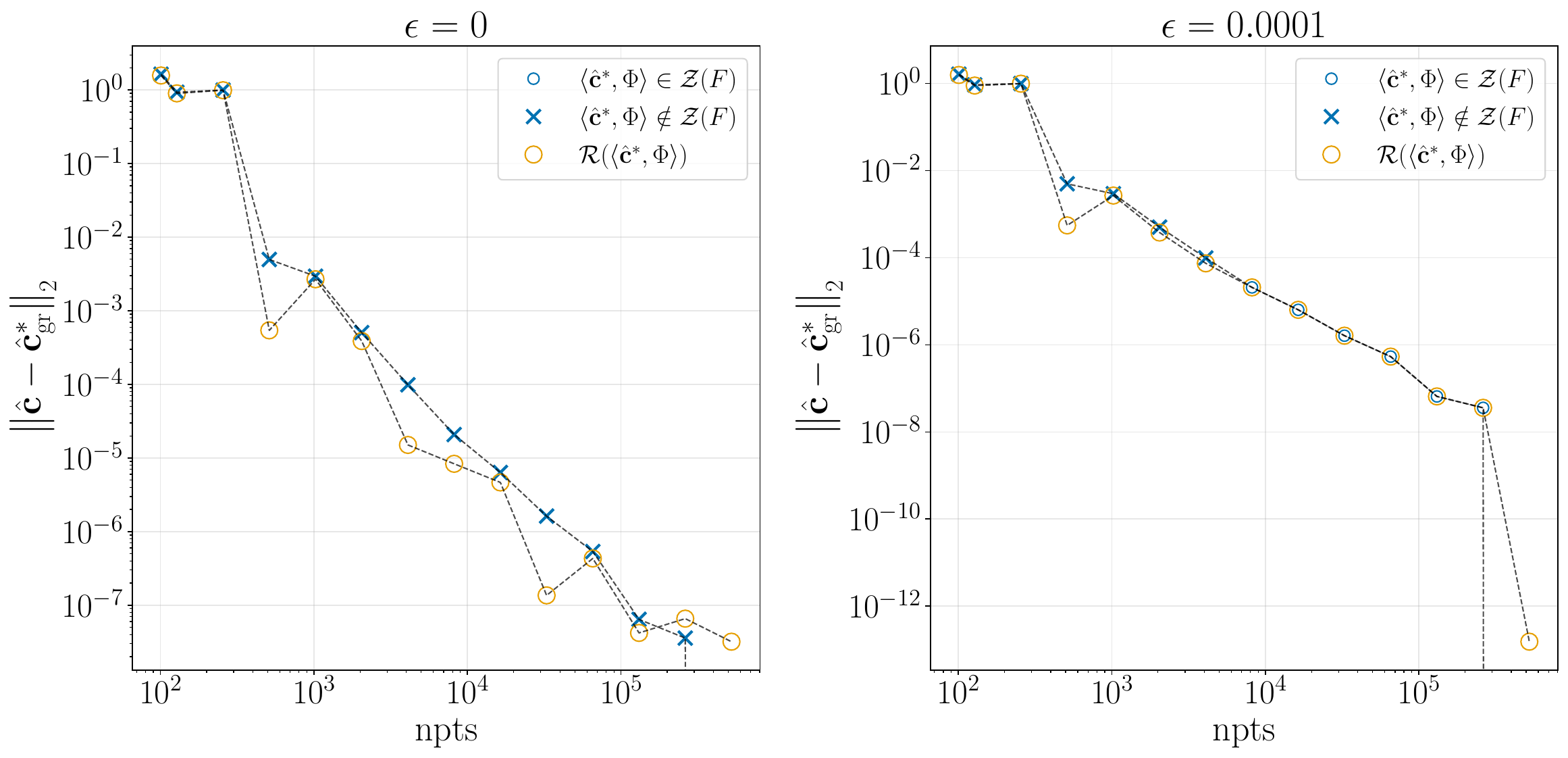}
    \vspace*{-10pt}
    \caption{Convergence of the retraction method for the target function $(1-\epsilon) g_{\strut \text{SV}}$ for different problem sizes of LSIP \cref{eq:lsip-prob} given by \texttt{npts} and $d=101$, for the two cases $\epsilon=0$ (left), which corresponds to the fully-coherent regime, and $\epsilon=10^{-4}$ (right). In both illustrations, $\hbc$ corresponds to the polynomial coefficients in the Chebyshev basis for the optimal polynomial $\inner{\hbc^\ast}$ (shown in blue), in which case $\hbc = \hbc^\ast$, and its retraction $\mathcal{R}(\inner{\hbc^\ast})$ (shown in maize), and $\hbc^\ast_{\text{gr}}$ is the optimal solution to the LSIP with $\texttt{npts}=2^{19}$, which is very close to the optimal solution of \cref{eq:intro-opt-prob}. Additionally illustrated on these plots are when $\inner{\hbc^\ast}$ is a feasible polynomial (circle) and not a feasible polynomial (cross).}
    \label{fig:polyspace_conv_unif_amp}
    \vspace*{-25pt}
\end{figure}
\begin{SCfigure}[][!ht]
    \includegraphics[width=0.55\linewidth]{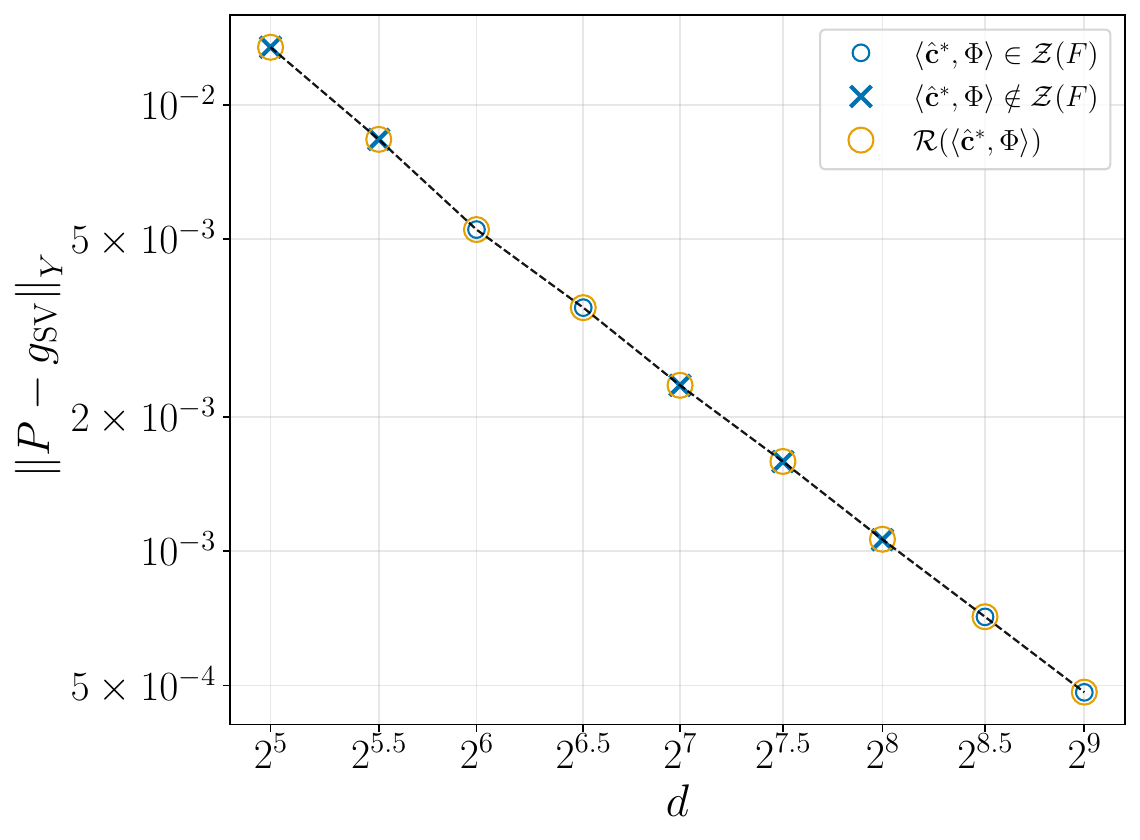}
  \caption{Plot of the approximation error $||P -  g_{\strut \text{SV}}||_Y$, for the two cases where (i) $P$ is the polynomial corresponding to the optimal solution $\hbc^\ast$ (shown in blue) to LSIP \cref{eq:lsip-prob}, and (ii) $P$ is the retraction $\mathcal{R}(\inner{\hbc^\ast})$ (shown in maize), versus the polynomial degree $d$. Each instance of the LSIP \cref{eq:lsip-prob} uses $\texttt{npts}=2^{19}$. Additionally shown are when $\inner{\hbc^\ast}$ is feasible (circle) and not feasible (cross).}
\label{fig:degree_scaling_amp}
\vspace*{-16pt}
\end{SCfigure}

In \Cref{fig:degree_scaling_amp}, we illustrate convergence to the target function $g_{\strut \text{SV}}$ while varying the degree $d$ (chosen to be odd integers, with fractional powers of two rounded up to the nearest odd integer), and we set \texttt{npts}$=2^{19}$. The \texttt{cvxpy} default solver \texttt{CLARABEL} was used for the numerics in this figure. We observe that as $d$ increases, the maximum point-wise error on $Y$ with respect to the target function progressively decreases, for both $\inner{\hbc^\ast}$ and $\mathcal{R}(\inner{\hbc^\ast})$, and moreover retraction does not increase the error.

\subsubsection{Performance for matrix inversion\label{subsubsubsec:mat_inv}}
Next, we apply the method to the target function $g_{\strut \text{MI}}$. We compare the retraction method with the strategy proposed in \cite{sunderhauf2025matrixinversionpolynomialsquantum} for enforcing the QSP feasibility constraints, starting from the unconstrained optimal polynomial for this problem (for odd $d$) derived in \cite{Berntson2024}. The strategy begins with the unconstrained optimal polynomial, which admits a closed-form analytical expression. If this polynomial violates the bound constraints, it is multiplied by a polynomial approximation to a windowing function, i.e., a function that is close to $0$ on $[0,1/\kappa)$ and close to $1$ on $[1/\kappa,1]$. Similar windowing strategies have previously been used in other works on QSP; see, e.g., \cite{Gily_n_2019, Martyn2021}.
\begin{figure}[!tp]
    \centering
    \includegraphics[width=\linewidth]{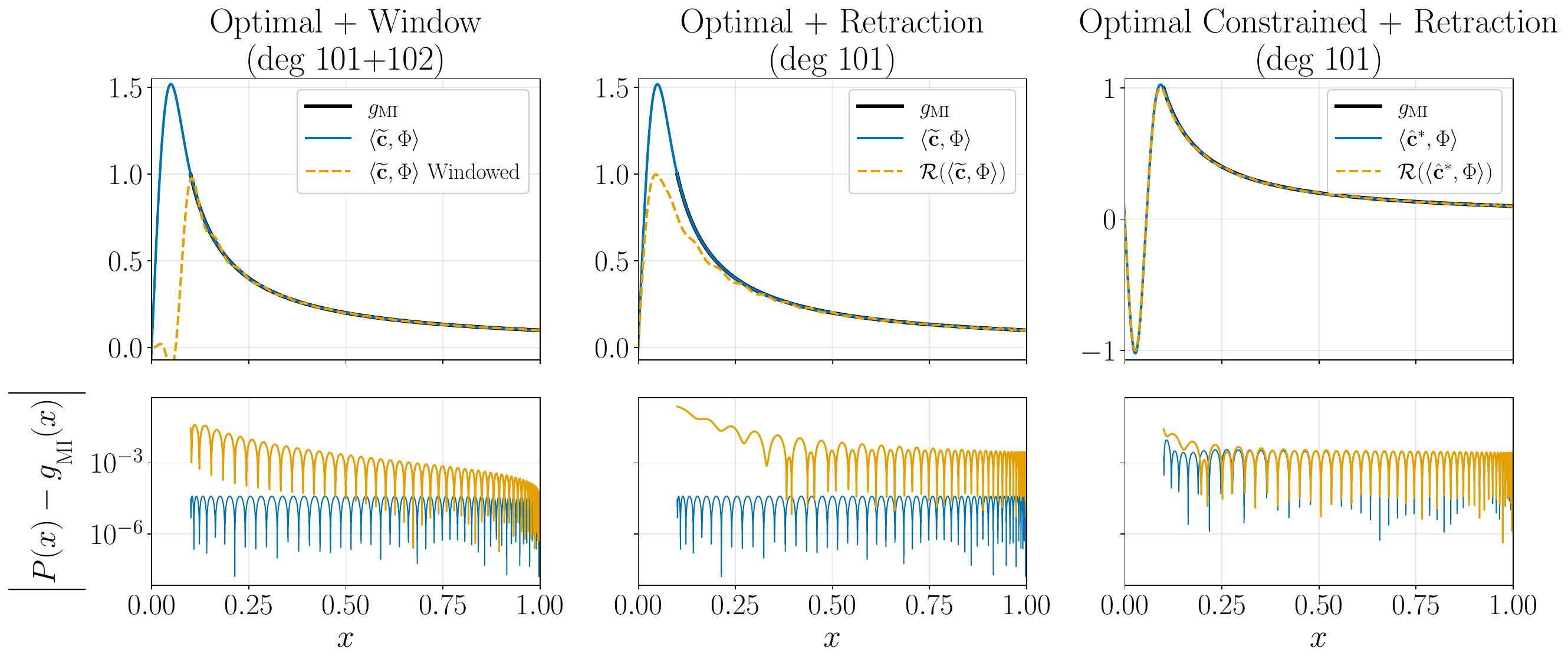}
    \vspace{-22pt}
    \caption{Comparison of three approaches for finding polynomial approximations to $g_{\strut \text{MI}}$. Left: the optimal unconstrained degree-$101$ polynomial $\inner{\widetilde{\bc}}$ from \cite{Berntson2024} (blue), with constraints enforced by multiplying it by a degree-$102$ even polynomial approximation of a windowing function (maize), as described in text. Center: the windowed polynomial of the left panel is replaced by the retraction $\mathcal{R}(\inner{\widetilde{\bc}})$ (maize) of the same optimal unconstrained polynomial. Right: a degree-$101$ polynomial $\inner{\hbc^\ast}$ (blue) obtained by solving LSIP \cref{eq:lsip-prob}, as described in the text, and its retraction $\mathcal{R}(\inner{\hbc^\ast})$ (maize). Bottom row: absolute point-wise errors relative to $g_{\strut \text{MI}}$ for the corresponding polynomials $P(x)$ in the top panel; e.g., in the bottom-center panel, $P$ denotes $\inner{\widetilde{\bc}}$ and $\mathcal{R}(\inner{\widetilde{\bc}})$.}
    \label{fig:mat_inv_ex}
    \vspace{-28pt}
\end{figure}

We set $d=101$, and show the comparisons in \Cref{fig:mat_inv_ex}. The left panel shows the strategy from \cite{sunderhauf2025matrixinversionpolynomialsquantum} where the degree-$101$ optimal unconstrained polynomial approximation to $g_{\strut \text{MI}}$, denoted $\inner{\widetilde{\bc}}$, is multiplied by a degree-$102$ even polynomial approximation of the windowing function, giving rise to a degree-$203$ windowed polynomial with same parity as $\inner{\widetilde{\bc}}$. Both $\inner{\widetilde{\bc}}$ and its windowed version are shown in the top-left panel. In the top-center panel, we use the same $\inner{\widetilde{\bc}}$ as in the left panel, but instead of the windowed polynomial, we plot the retraction $\mathcal{R}(\inner{\widetilde{\bc}})$ of $\inner{\widetilde{\bc}}$. In the top-right panel, we plot the polynomial $\inner{\hbc^\ast}$ obtained by solving the LSIP \cref{eq:lsip-prob} with $\texttt{npts}=128$, and its retraction $\mathcal{R}(\inner{\hbc^\ast})$. In each of the top panels, we have also displayed $g_{\strut \text{MI}}$ for reference. The bottom panels show the plots of the absolute point-wise error over $Y$, between $g_{\strut \text{MI}}$ and each of the two polynomials in the corresponding top panel (except $g_{\strut \text{MI}}$). The center panel shows that the point-wise approximation error to the target function of the nonlinear retraction polynomial (of degree $101$) is larger than that of the windowed polynomial shown in the left panel. This difference in performance is expected, since the windowed polynomial has twice the degree. In the right panel, $\mathcal{R}(\inner{\hbc^\ast})$ demonstrates lower error for $x \leq 0.4$, in comparison to the center panel, and even outperforms the degree-$203$ windowed polynomial from the left panel.
\section{Discussions}
\label{sec:discussions}

We study numerical methods for constrained polynomial approximation in QSP near the fully-coherent regime, where an optimal approximant reaches the boundary of the feasible set. Bound constraints enforced on a finite grid do not control the polynomial between grid points, so a discretized LSIP can return an infeasible candidate despite satisfying every sampled constraint. A QSP representation requires global boundedness on the full interval, so residual violations must be certified and corrected.

The classical Remez exchange method offers an efficient alternative to dense discretization. Our Remez- and active-set-based heuristic performs well for many tested instances, but its updates can become numerically unstable. The most reliable workflow in our experiments proceeds adaptively. We solve a moderately discretized LSIP, certify the global extrema of the candidate polynomial, and apply nonlinear Fourier retraction only if the candidate remains infeasible. Retraction preserves the polynomial degree and adds only modest classical post-processing.
Nonlinear Fourier analysis provides a bridge between the constrained approximation problem and QSP representations. Note that even though we introduced nonlinear Fourier retraction as a tool to obtain feasible polynomials from nearly feasible solutions to \cref{eq:intro-opt-prob}, obtained via solving \cref{eq:lsip-prob}, we could also use retraction to replace Lines~13-15 of \Cref{alg:constrained_remez_qsp}, to enforce constraint feasibility of that algorithm.

Both algorithms are implemented in the Python version of \texttt{qsppack}. Across the problems and parameter regimes tested, we find that retraction robustly enforced feasibility while largely preserving approximation accuracy. Future works include bounding the distortion introduced by retraction near the fully-coherent boundary with certified constraint violation and retraction accuracy. 

\section*{Notes}
\label{sec:notes}

Toward the end of this work, we became aware of a constrained approximation method due to Foucart and Powers \cite{foucart2017basc}. Using the Fej\'er-Riesz theorem, the bound constraints $-1 \le \langle \bc, \Phi(x) \rangle \le 1$ in the LSIP \cref{eq:intro-opt-prob} can be represented by semidefinite matrix constraints. We leave a detailed study of the performance of this alternative formulation, especially in the fully-coherent regime, for future work.

\bibliographystyle{siamplain}
\bibliography{references}

\appendix
\section{Omitted proofs and other supporting results}
\label{app:additional}

\subsection{Best approximation examples and existence proof}
\label{appssec:additional}

We provide some examples below where we have unique and non-unique best approximation polynomials. Following \cite{dunham1974families}, the set $\{\phi_1,\dots,\phi_{\tilde{d}}\}$ is a Chebyshev set on $(X,\{0\})$ (resp. $(X,\emptyset)$) when $d$ is odd (resp. even): for a closed subset $V \subseteq X$, with $C(X,V)$ denoting the subspace of $C(X)$ vanishing on $V$, a subset $\{\psi_1,\dots,\psi_{\tilde{d}}\} \subseteq C(X,V)$ is called a \textit{Chebyshev set} on $(X, V)$ if no nontrivial linear combination of them vanishes on $\tilde{d}$ distinct points on $X \setminus V$.
\begin{ex}[Unique best approximations]
\begin{enumerate}[(a)]
    \item Let $Y \subseteq (0,1]$ be any compact subset, such that $|Y| \ge \tilde{d}$, and let $g \in \calZ(F)$. Then there is a unique best approximation $\bc^\ast$ such that $\langle \bc^\ast, \Phi \rangle = g$. The reason is that, the set of Chebyshev polynomials $\{\phi_1,\dots,\phi_{\tilde{d}}\}$ forms a Chebyshev set on $(X,\emptyset)$ or $(X,\{0\})$, as mentioned earlier.
    \item Let $Y = X$, let $g(x) = \kappa x^2$, for some $\kappa > 0$, and let $d = 1$. In this case, it is easy to see that the best approximation polynomial must be of the form $c_1^\ast x$, for some $0 \le c_1^\ast \le \min\{1,\kappa\}$. We claim that $c_1^\ast$ is unique. To see this, a simple calculation shows that $\sup_{x \in Y} |c_1^\ast x - \kappa x^2| = \max\left\{\frac{(c_1^\ast)^2}{4 \kappa}, \kappa - c_1^\ast\right\}$, and the right-hand side has a unique minimizer as a function of $c_1^\ast$, over $[0, \min\{1,\kappa\}]$.
\end{enumerate}
\end{ex}
\begin{ex}[Non-unique best approximations]
\begin{enumerate}[(a)]
    \item Let $Y = \{0\}$. In this case, if $d$ is odd, then every $\bc^\ast \in F$ is a best approximation to $g$, while if $d$ is even, then the set of best approximations to $g$ is given by $\{\bc^\ast \in F: c^\ast_1 = \min \{ \max \{g(0),-1\}, 1\}\}$, and this is an affine subspace of dimension $\tilde{d} - 1$.
    \item Let $Y = [0,a]$, for some $a \le 1$, let $g(x)= C$ for all $x \in Y$, and for some constant $C \ge \frac{1}{2}$, and suppose $d$ is odd. In this case, the set of best approximations to $g$ is given by $\{\bc^\ast \in F : \innerx{\bc^\ast} \ge 0, \forall \; x \in Y\}$, which is a closed, convex set.
\end{enumerate}
\end{ex}

Next, we present the proof of \Cref{lem:existence-optimal-sol}.

\begin{proof}[Proof of \texorpdfstring{\Cref{lem:existence-optimal-sol}}{}]
Define the set $\hatZ = \hatX \cup \hatY$ if $d$ is even, and $\hatZ = (\hatX \cup \hatY) \setminus \{0\}$ if $d$ is odd. We will split the proof of existence of an optimal solution into two cases: (i) $|\hatZ| > \tilde{d}$, and (ii) $|\hatZ| \le \tilde{d}$. We recall that the LSIP~\cref{eq:lsip-prob} is equivalent to the minimization of $\norm{\inner{\bc} - g}{\hatY}$ subject to $\norm{\inner{\bc}}{\hatX} \le 1$, and we denote the feasible set of the latter as $\tilde{F}$ for this proof. Note that $\tilde{F}$ is closed.
\begin{enumerate}[(i)]
    \item For $R >0$, we define the closed ball $B_R := \{\bc \in \RR^{\tilde{d}} : \norm{\bc}{\infty} \le R\}$, and let $\partial B_R := \{\bc \in \RR^{\tilde{d}} : \norm{\bc}{\infty} = R\}$ denote its boundary. Then for any $\bc \in \partial B_1$, the polynomial $\inner{\bc}$ is not zero at all points in $\hatZ$ (since $|\hatZ| > \tilde{d}$), which implies $\mu := \inf_{\bc \in \partial B_1} \norm{\inner{\bc}}{\hatZ} > 0$, and choose $R>0$ such that $R\mu = \max\{1,2\norm{g}{\hatY}\}$. Now let $\bc' \in \RR^{\tilde{d}} \setminus B_R$. Then there exists $x \in \hatZ$ such that $|\innerx{\bc'}| > R\mu$. If $x \in \hatX$, this implies $\norm{\inner{\bc'}}{\hatX} > 1$, so that $(\bc',e) \not \in F_{(\hatX,\hatY)}$ for any $e$. If $\hatY \ni x \not \in \hatX$, then  notice that by the triangle inequality we have $\norm{\inner{\bc'}-g}{\hatY} \ge \norm{\inner{\bc'}}{\hatY} - \norm{g}{\hatY} > \norm{g}{\hatY}$. Therefore, in both cases we get
    \vspace*{-0.2cm}
    \begin{equation*}
        \inf_{\bc \in \tilde{F}} \norm{\inner{\bc} - g}{\hatY} =  \inf_{\bc \in \tilde{F} \cap B_R} \norm{\inner{\bc} - g}{\hatY},
        \vspace*{-0.3cm}
    \end{equation*}
    and since $\tilde{F} \cap B_R$ is compact, this infimum is attained for some $\hbc^\ast$. An optimal solution to the LSIP~\cref{eq:lsip-prob} is then given by $(\hbc^\ast, \norm{\inner{\hbc^\ast}-g}{\hatY})$.
    \item In this case, define a function on $x \in \hatZ$ as follows:
    \vspace*{-0.2cm}
    \begin{equation*}
        h(x) :=
        \begin{cases}
            0 & \quad \text{if } x \in \hatZ \setminus \hatY \\
            \min \{ \max \{g(x), -1\}, 1\} & \quad \text{if } x \in \hatY \cap \hatX \cap \hatZ \\
            g(x) & \quad \text{if } x \in (\hatY \cap \hatZ) \setminus \hatX.
        \end{cases}
        \vspace*{-0.3cm}
    \end{equation*}
    Then there exists $\hbc^\ast \in \RR^{\tilde{d}}$ such that $\innerx{\hbc^\ast} = h(x)$ for all $x \in \hatZ$, and $(\hbc^\ast, \norm{\inner{\hbc^\ast}-g}{\hatY})$ is an optimal solution to the LSIP \cref{eq:lsip-prob}.
\end{enumerate}
\end{proof}

\subsection{Mapping between \texorpdfstring{$P(x)$}{} and \texorpdfstring{$b(z)$}{}}
\label{appssec:coefmap}

Here we assume that we have a polynomial $P \in \RR[x]$ of degree $d$ and parity $d \bmod{2}$, which is obtained as an approximation to the target function $g \in C(Y)$. Let $P$ have the Chebyshev expansion $\sum_{j=1}^{\tilde{d}} c_j \phi_j(x)$, and suppose $\norm{P}{X} \le 1$. We want to find the Laurent polynomial $b(z)$, as explained in \Cref{ssec:qsp-nlft} and defined by $b(e^{2i\theta}) := e^{id \theta} P(\cos \theta)$, that can be subsequently used to find the QSP phase factors using the QSP-NLFT correspondence. When $d$ is even, this leads to $b(e^{2i\theta}) = e^{id\theta} \sum_{j=1}^{\tilde{d}} c_{j} T_{2j-2} (\cos \theta) = \sum_{j=1}^{\tilde{d}} \frac{c_{j}}{2} \left( e^{2i\theta(\frac{d}{2} + j-1)} + e^{2i\theta(\frac{d}{2} - j+1)} \right)$, or equivalently,
\vspace*{-0.3cm}
\begin{equation}
\label{eq:evenmap}
b(z) = \sum_{j=1}^{\tilde{d}} \frac{c_{j}}{2} \left( z^{\frac{d}{2} + j - 1} + z^{\frac{d}{2} -j + 1} \right) = \frac{1}{2}\sum_{j=0}^{\frac{d}{2}-1} c_{\frac{d}{2}+1-j} z^j + c_1 z^{\frac{d}{2}} + \frac{1}{2}\sum_{j=1}^{d/2} c_{j+1} z^{\frac{d}{2}+j}. \vspace*{-0.3cm}
\end{equation}
Similarly, for odd $d$, a calculation shows that
\vspace*{-0.3cm}
\begin{equation}
\label{eq:oddmap}
    b(z) = \frac{1}{2}\sum_{j=0}^{\frac{d-1}{2}} c_{\frac{d+1}{2}-j} z^j + \frac{1}{2}\sum_{j=0}^{\frac{d-1}{2}} c_{j+1} z^{\frac{d+1}{2}+j}.
\end{equation}

\subsection{Proof of \texorpdfstring{\Cref{thm:bounds-astar-real-imag}}{}}
\label{appssec:proof-main-theorem}
In the following proof, we use the notations from \Cref{ssec:alg-retraction} and \Cref{thm:bounds-astar-real-imag}.
\begin{proof}
We get $\hat{R}_0 = (4 \pi)^{-1} \int_{0}^{2 \pi} \log \left| 1 - |b(z)|^2 \right| d\theta + \frac{i \pi}{2} |\TT_{+}|$, by the Fourier series formula. Then the expressions for $\Re(a^\ast(0))$ and $\Im(a^\ast(0))$, in \cref{eq:astar-formulas} follow, by noting $a^\ast(0) = e^{G(0)} = e^{\hat{R}_0}$. Note that $K$ is finite as $\log \left| 1 - |b(\cdot)|^2 \right| \in L^1(\TT)$.

For the $L^p$ statement, we will prove something stronger, that $|a^\ast|$ is a continuous function on $\TT$. We note that $|a^\ast| = e^{\Re(G)}$, and a simple computation shows that $\Re(G) = \frac{1}{2}\log \left| 1-|b(z)|^2 \right| - \frac{\pi}{2}\mathcal{H}\left( \mathbf{1}_{\TT_{+}}\right)$, where $\mathbf{1}_{\TT_{+}}$ denotes the indicator function on $\TT_{+}$, and therefore, we get $|a^\ast(z)| = \sqrt{\left| 1-|b(z)|^2 \right|} e^{-\frac{\pi}{2} \mathcal{H}\left(\mathbf{1}_{\TT_{+}}\right)}$. Under the assumptions on the polynomial $b$, it is clear that $\TT_{+}$ is a finite union of disjoint open intervals of $\TT$; so let us assume that $\TT_{+} = \bigcup_{j=1}^{L} \TT_{(\alpha_j,\beta_j)}$, where $\TT_{(\alpha_j,\beta_j)} := \{e^{i\theta}: \theta \in (\alpha_j, \beta_j)\}$, for some $0 \le \alpha_j < \beta_j \le 2 \pi$. It is always possible to write this, because if there is any interval of the form $\{e^{i\theta}: 0 \le \theta < \beta_j\} \cup \{e^{i\theta}: \alpha_j < \theta \le 2\pi\}$ for some $0< \beta_j < \alpha_j < 2\pi$, then the indicator function on this set agrees a.e. with the sum of the indicator functions on $\TT_{(0,\beta_j)}$ and $\TT_{(\alpha_j, 2\pi)}$. Thus the Hilbert transform of $\mathbf{1}_{\TT_{+}}$ has the very simple closed form $\mathcal{H}\left( \mathbf{1}_{\TT_{+}} \right)(e^{i\zeta}) = \frac{1}{\pi} \sum_{j=1}^{L} \log \left| \frac{\sin((\zeta - \alpha_j)/2)}{\sin((\zeta - \beta_j)/2)} \right|$. Combining everything, we get $|a^\ast(e^{i\theta})| = \left| \left( 1-|b(e^{i\theta})|^2 \right) \prod_{j=1}^{L} \frac{\sin((\theta - \beta_j)/2)}{\sin((\theta - \alpha_j)/2)} \right|^{\frac{1}{2}}$, for $e^{i\theta} \in \TT$. 

Now, we simplify the product $\prod_{j=1}^{L} \frac{\sin((\theta - \beta_j)/2)}{\sin((\theta - \alpha_j)/2)}$ by canceling all the terms which are common to both the numerator and the denominator, and suppose the denominator in the result has the form $\digamma := \prod_{j=1}^{L'} \sin((\theta - \alpha_j)/2)$. Then for each $\alpha_j$ (note that this set of $\alpha_j$ are discrete) appearing in this product, it must be the case that $(1 - |b(e^{i\alpha_j})|^2) = 0$. Fix such an $\alpha_j$. Since $1 - |b(z)|^2 = 1 - b(z)b^\ast(z)$, for $z \in \TT$, and $1 - b(z)b^\ast(z) = h(z)/z^s$, where $h$ is a polynomial, and $s \in \mathbb{N}$ is the degree of $b$, we conclude that $h(e^{i \alpha_j})=0$. Therefore, $h(z) = (z - e^{i\alpha_j}) \tilde{h}(z)$, where $\tilde{h}$ is another polynomial, and this implies $\frac{\left( 1-|b(e^{i\theta})|^2 \right)}{\sin((\theta - \alpha_j)/2)} = \left(\frac{e^{i\theta} - e^{i\alpha_j}}{\sin((\theta - \alpha_j)/2)}\right) \tilde{h}(e^{i\theta}) = 2i e^{i(\theta + \alpha_j)/2} \tilde{h}(e^{i\theta})$, which is a continuous function on $\TT$. Restricting to a small neighborhood of $e^{i\alpha_j}$, we then conclude that $|a^\ast|$ is continuous at $e^{i\alpha_j}$. Repeating this argument for every $\alpha_j$ appearing in the product $\digamma$, gives that $|a^\ast|$ is continuous on $\TT$.
\end{proof}

\section{Details for the Remez- and active-set-based method}
\label{app:remez-active-set}

In this subsection, we summarize the numerical implementation details used in the experiments for \cref{alg:constrained_remez_qsp}. We slightly engineer \cref{alg:constrained_remez_qsp} to increase its numerical stability. In particular, we adopt two strategies: (i) target amplitude rescaling, and (ii) approximation subinterval shrinkage . For (i), we set the approximation target in \cref{alg:constrained_remez_qsp} as $(1 - \xi) g(x)$ for some small value $0< \xi < 1$. This strategy can stabilize the process and avoid the active set dominating the iterative process, especially when the target function exactly equals one on a subinterval like with the threshold projector target function. The approximation error is still evaluated using the original unscaled target function. For (ii), we set the approximation interval to a smaller subset with a separation of size $\eta$, i.e., $Y = [a, b] \mapsto Y^\prime = [a + \eta, b - \eta]$. This strategy stabilizes the process when the function may need a sharp reflection at the boundary value like with the uniform singular value amplification target function. 

In practice, we find these two tricks can enhance the numerical stability in some cases. However, for uniform singular value amplification, there does not appear to be a problem-agnostic parameter choice that consistently ensures convergence. The parameters used in the numerical experiments are:
\begin{itemize}
    \item Threshold projector: $\xi = 10^{-6}$ and $\eta = 0$.
    \item Matrix inversion: $\xi = 10^{-6}$ and $\eta = 0$.
    \item Uniform singular value amplification: $\xi = 0$ and $\eta = 1.33 \times 10^{-4}$.
\end{itemize}

\section{Nonlinear Fourier retraction for \texorpdfstring{$g_{\strut \text{MI}}$ and $g_{\strut \text{TP}}$}{}}
\label{app:retraction}
In this section, we repeat the numerical experiments analogous to \Cref{fig:polyspace_conv_unif_amp,fig:degree_scaling_amp}, for the matrix inversion and the threshold projector problems, thus showcasing the consistent reliability of nonlinear Fourier retraction across QSP use cases.

\subsection{Matrix inversion (extra results)}
\label{appssec:retract_exp_matrix}

In \Cref{fig:polyspace_conv_mat_inv}, we demonstrate the performance of retraction as \texttt{npts} is increased by repeating exactly the same process as for \Cref{fig:polyspace_conv_unif_amp}, with the only change being that we replace $g_{\strut \text{SV}}$ with $g_{\strut \text{MI}}$ throughout. We use the same degree, $d=101$, as in \Cref{fig:polyspace_conv_unif_amp}; moreover, all other quantities and variables, such as $\texttt{npts}=2^{19}$ and $\epsilon=10^{-4}$, remain unchanged. The computation of $\hbc^\ast_{\text{gr}}$ is also performed in the same way, after replacing $g_{\strut \text{SV}}$ with $g_{\strut \text{MI}}$.
\begin{figure}[!hp]
\vspace*{-12pt}
    \centering
    \includegraphics[width=0.8\linewidth]{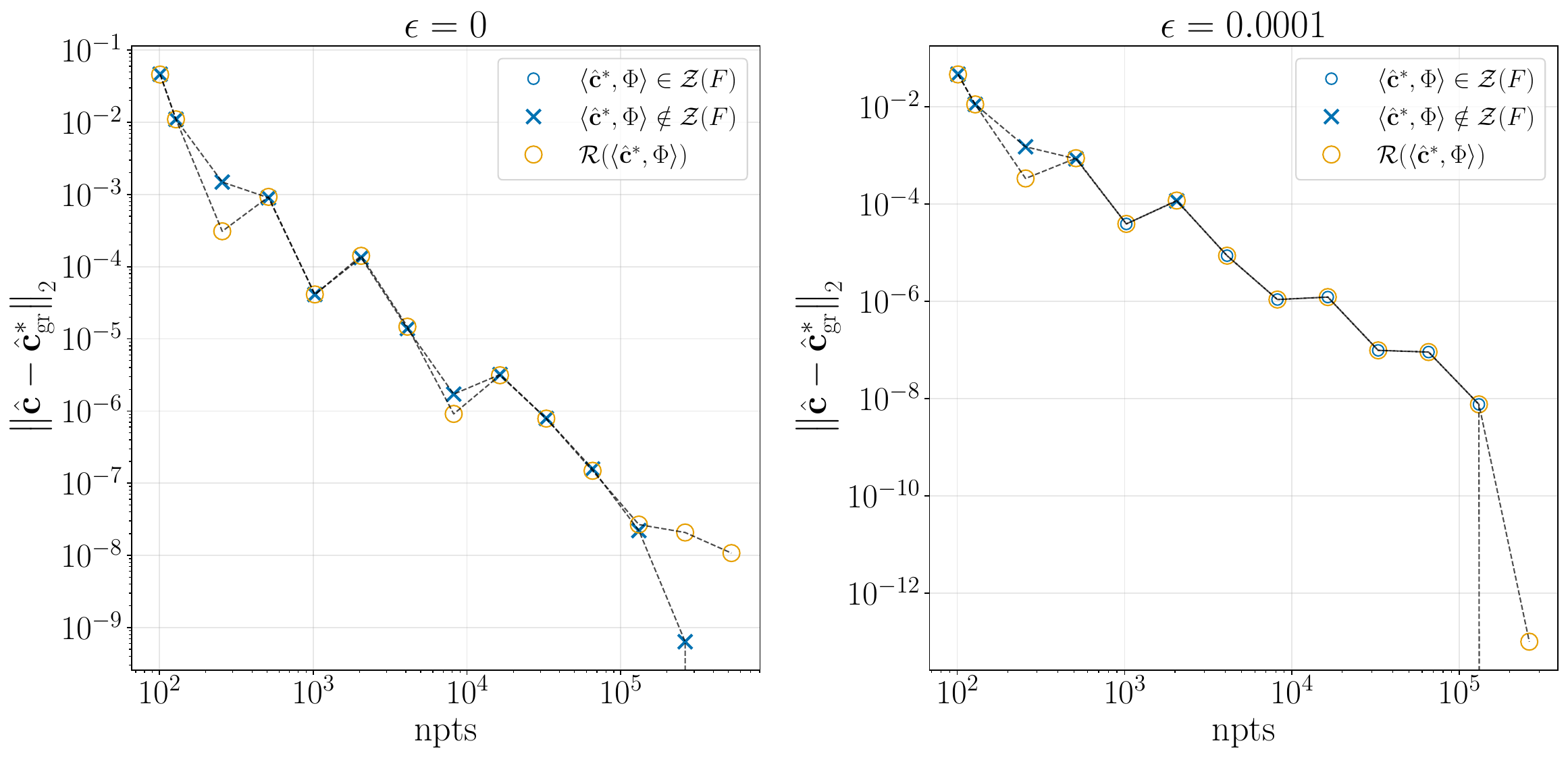}
    \vspace*{-10pt}
    \caption{Convergence of the retraction method for the target function $(1-\epsilon) g_{\strut \text{MI}}$ for different problem sizes of LSIP \cref{eq:lsip-prob} given by \texttt{npts} and $d=101$, for the two cases $\epsilon=0$ (left), and $\epsilon=10^{-4}$ (right). Similar to \Cref{fig:polyspace_conv_unif_amp}, $\hbc$ corresponds to the polynomial coefficients in the Chebyshev basis for $\inner{\hbc^\ast}$ (shown in blue), and its retraction $\mathcal{R}(\inner{\hbc^\ast})$ (shown in maize), and $\hbc^\ast_{\text{gr}}$ is the optimal solution to the LSIP with $\texttt{npts}=2^{19}$. Additionally illustrated are when $\inner{\hbc^\ast}$ is a feasible polynomial (circle) and not a feasible polynomial (cross).}
    \label{fig:polyspace_conv_mat_inv}
\vspace*{-28pt}
\end{figure}
\begin{SCfigure}[][!hp]
    \includegraphics[width=0.55\linewidth]{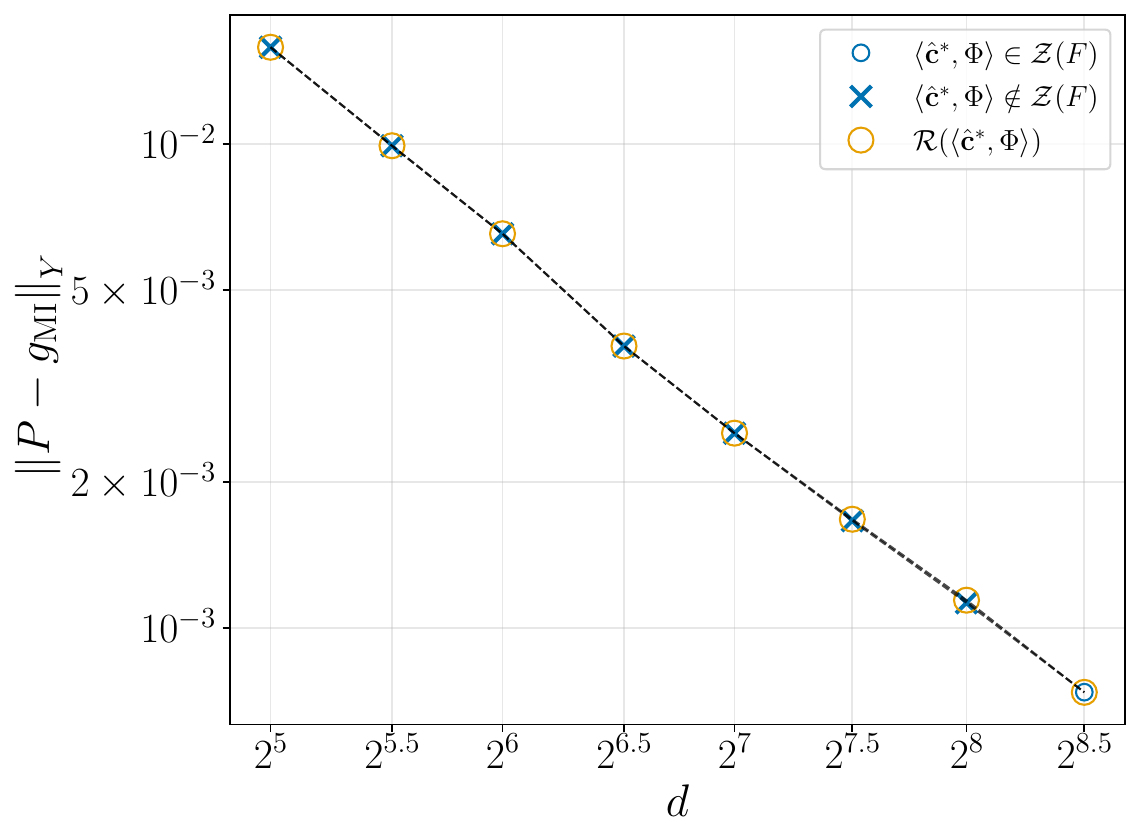}
    \caption{Plot of the approximation error $||P -  g_{\strut \text{MI}}||_Y$, for the two cases where (i) $P$ is the polynomial corresponding to the optimal solution $\hbc^\ast$ (shown in blue) to LSIP \cref{eq:lsip-prob}, and (ii) $P$ is the retraction $\mathcal{R}(\inner{\hbc^\ast})$ (shown in maize), versus the polynomial degree $d$. Each instance of the LSIP \cref{eq:lsip-prob} uses $\texttt{npts}=2^{19}$. Additionally shown are when $\inner{\hbc^\ast}$ is feasible (circle) and not feasible (cross).}
    \label{fig:degree_scaling_mat_inv}
\end{SCfigure}

In \Cref{fig:degree_scaling_mat_inv}, we illustrate the decrease in the supremum norm on $Y$ of the difference of $g_{\strut \text{MI}}$ and both the optimal solution to the LSIP \cref{eq:lsip-prob}, given by $\inner{\hbc^\ast}$, and its retraction $\mathcal{R}(\inner{\hbc^\ast})$, along with the feasibility of these polynomials. We choose $\texttt{npts}=2^{19}$, and the other details are kept the same as for \Cref{fig:degree_scaling_amp}. We see that after retraction of the polynomials, this error does not increase, and this effect is seen for $d \le 364$, as the polynomials $\inner{\hbc^\ast}$ are infeasible. For $d > 364$, retraction does not change the polynomial as $\inner{\hbc^\ast}$ is already feasible.

\subsection{Threshold projector}
\label{appssec:retract_exp_threshold}

Here we demonstrate the success of the nonlinear Fourier retraction algorithm for the threshold projector target function $g_{\strut \text{TP}}$. This target function is displayed in \Cref{fig:thresh_proj_inset}, along with the polynomial $\inner{\hbc^\ast}$ corresponding to the optimal solution to LSIP \cref{eq:lsip-prob}, with $d=128$ and $\texttt{npts}=256$, and also its nonlinear retraction $\mathcal{R}(\inner{\hbc^\ast})$. To illustrate the importance of the retraction method for this example, a zoomed-in inset portrays the numerous constraint violations of $\inner{\hbc^\ast}$ (since we are in the fully-coherent regime) in \Cref{fig:thresh_proj_inset}, and shows that the retracted polynomial remains close to the target function while rigidly enforcing the constraints.
\begin{SCfigure}[][!ht ]
    \includegraphics[width=0.55\textwidth]{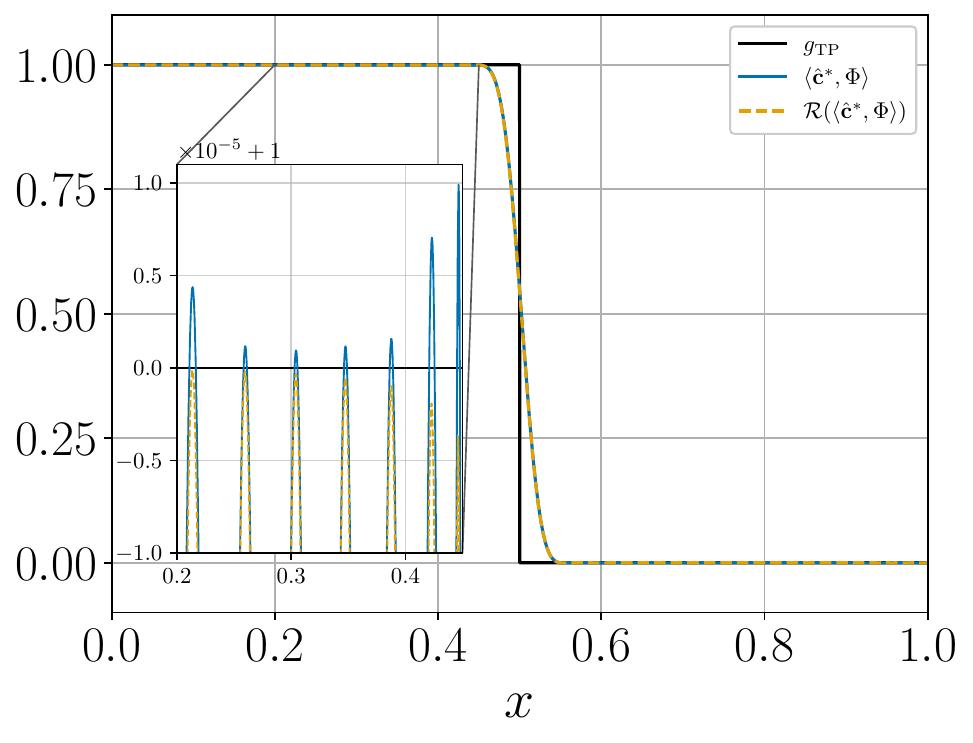}
    \caption{Example of degree-$128$ approximation of $g_{\strut \text{TP}}$, found by solving the LSIP \cref{eq:lsip-prob} with \texttt{npts}$=256$. $g_{\strut \text{TP}}$ is shown in black, the optimal polynomial solution $\inner{\hbc^\ast}$ of \cref{eq:lsip-prob} is shown in blue, and the retraction $\mathcal{R}(\inner{\hbc^\ast})$ is shown in dashed maize. The inset zooms in on the fully-coherent regime from $x=0.2$ to $x=0.45$, illustrating the constraint violations of $\inner{\hbc^\ast}$.}
    \label{fig:thresh_proj_inset}
\vspace*{-16pt}
\end{SCfigure}
\begin{figure}[!hp]
\vspace*{-10pt}
    \centering
    \includegraphics[width=0.8\linewidth]{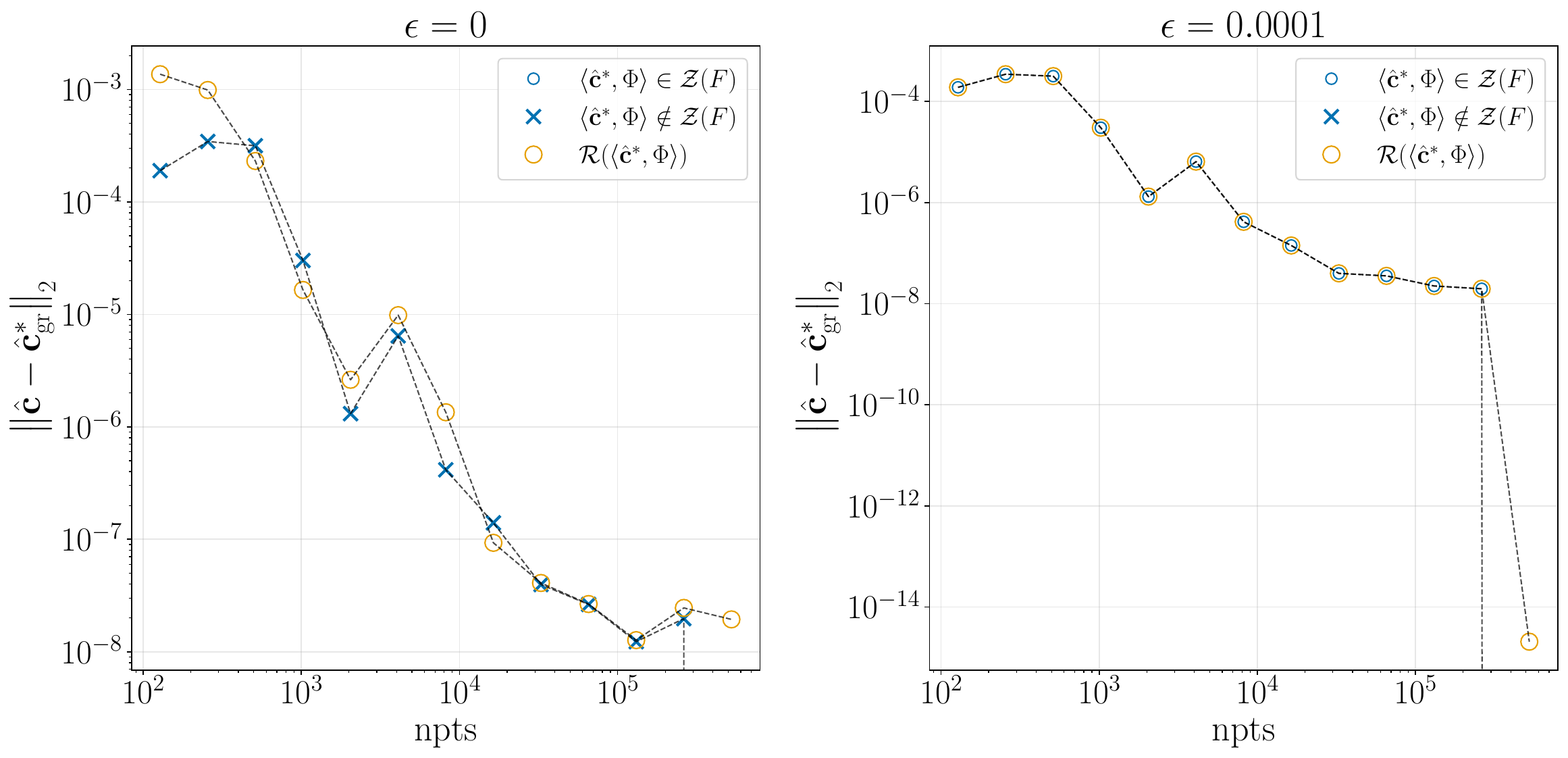}
    \vspace*{-10pt}
    \caption{Convergence of the retraction method for the target function $(1-\epsilon) g_{\strut \text{TP}}$ for different problem sizes of LSIP \cref{eq:lsip-prob} given by \texttt{npts} and $d=128$, for the two cases $\epsilon=0$ (left), and $\epsilon=10^{-4}$ (right). Similar to \Cref{fig:polyspace_conv_unif_amp}, $\hbc$ corresponds to the polynomial coefficients in the Chebyshev basis for $\inner{\hbc^\ast}$ (shown in blue), and its retraction $\mathcal{R}(\inner{\hbc^\ast})$ (shown in maize), and $\hbc^\ast_{\text{gr}}$ is the optimal solution to the LSIP with $\texttt{npts}=2^{19}$. Additionally illustrated are when $\inner{\hbc^\ast}$ is a feasible polynomial (circle) and not a feasible polynomial (cross).}
    \label{fig:polyspace_conv_thresh_proj}
    \vspace*{-16pt}
\end{figure}

The result of repeating the numerical experiment corresponding to \Cref{fig:polyspace_conv_unif_amp}, but for the target function $g_{\strut \text{TP}}$ is shown in \Cref{fig:polyspace_conv_thresh_proj}. A few other parameters that are different, as compared to \Cref{fig:polyspace_conv_unif_amp}, is that we set the degree to $d=128$, and $N=2^{15}$ in \textsc{MWeiss}. 

In \Cref{fig:degree_scaling_thresh_proj}, we show the result of repeating the numerical experiment from \Cref{fig:degree_scaling_amp}, but for the $g_{\strut \text{TP}}$ target function. For this fully-coherent regime example, we fix \texttt{npts}$=2^{19}$, and the degree parameter $d$ is varied over even integers. We see that the polynomials $\inner{\hbc^\ast}$, corresponding to the solution of the LSIP \cref{eq:lsip-prob}, are infeasible for all $d$, except $d=256$. But their nonlinear retraction $\mathcal{R}(\inner{\hbc^\ast})$ achieves almost the same error on $Y$ in the supremum norm, relative to $g_{\strut \text{TP}}$, but they are feasible polynomials. However, unlike \Cref{fig:degree_scaling_amp} and \Cref{fig:degree_scaling_mat_inv}, we observe super-polynomial scaling in the error.

\begin{SCfigure}[][!ht]
    \includegraphics[width=0.55\textwidth]{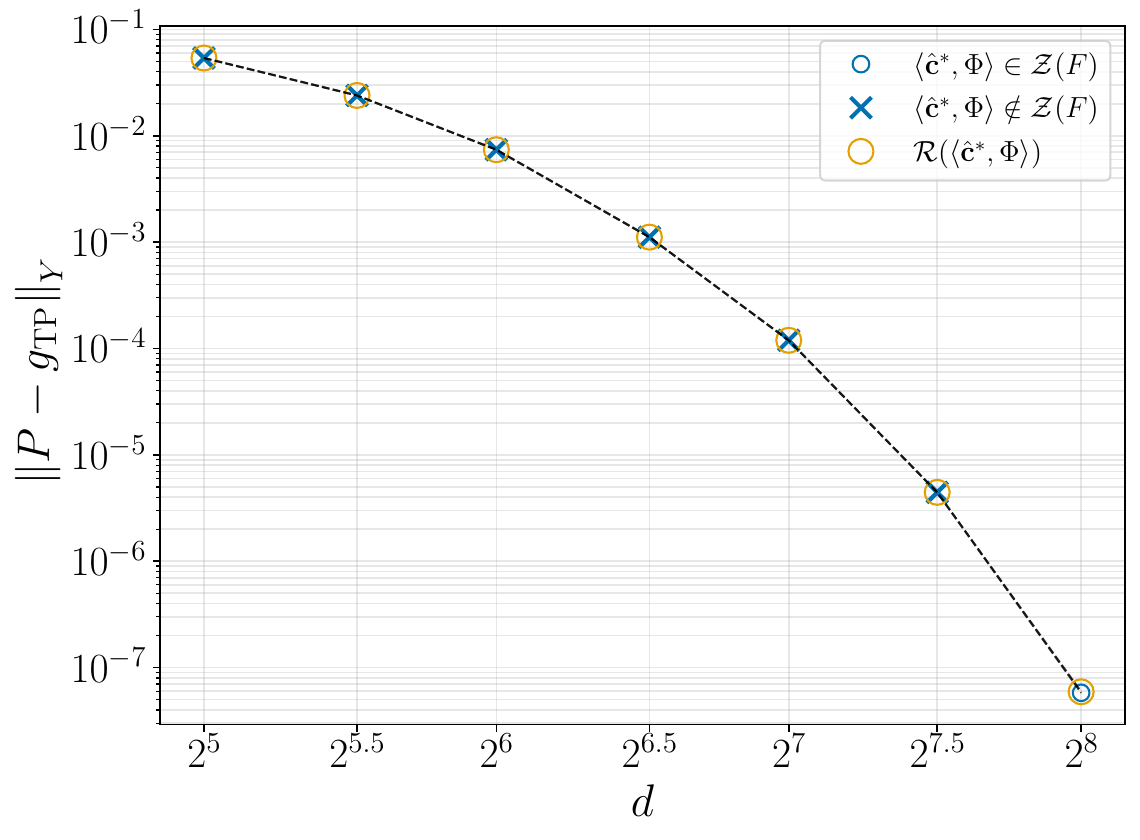}
    \caption{Plot of the approximation error $||P -  g_{\strut \text{TP}}||_Y$, for the two cases where (i) $P$ is the polynomial corresponding to the optimal solution $\hbc^\ast$ (shown in blue) to LSIP \cref{eq:lsip-prob}, and (ii) $P$ is the retraction $\mathcal{R}(\inner{\hbc^\ast})$ (shown in maize), versus the polynomial degree $d$. Each instance of the LSIP \cref{eq:lsip-prob} uses $\texttt{npts}=2^{19}$. Additionally shown are when $\inner{\hbc^\ast}$ is feasible (circle) and not feasible (cross).}
\label{fig:degree_scaling_thresh_proj}
\vspace*{-16pt}
\end{SCfigure}

\end{document}